\documentclass[a4paper,11pt,fleqn]{article}

\usepackage[T1]{fontenc}
\usepackage[utf8]{inputenc}
\usepackage[english]{babel}
\usepackage[margin=23mm,headheight=14pt,headsep=7mm,footskip=11mm]{geometry}
\usepackage{newtxtext}
\usepackage{amsmath,amssymb,amsthm}
\usepackage{newtxmath}
\usepackage{microtype}
\usepackage{graphicx}
\usepackage{booktabs,makecell,multirow,array,colortbl,dcolumn}
\usepackage{etoolbox,xspace,xstring}
\usepackage{xcolor}
\usepackage{bm}
\usepackage{float}
\usepackage{comment}
\usepackage{enumerate}
\usepackage{dsfont}
\usepackage{algorithm}
\usepackage{algpseudocode}
\usepackage[authoryear,longnamesfirst]{natbib}
\usepackage{caption}
\usepackage{fancyhdr}
\usepackage[hidelinks]{hyperref}
\hypersetup{
    pdftitle={Regression models for ordinal compositional data},
    pdfauthor={Beniamino Cappelletti-Montano, Monica Musio, Nicola Piras}
}

\fancypagestyle{plain}{%
  \fancyhf{}%
  \fancyfoot[C]{\thepage}%
}

\AtBeginEnvironment{table}{\small}

\theoremstyle{plain}

\newtheorem{theorem}{Theorem}[section]
\newtheorem{lemma}[theorem]{Lemma}
\newtheorem{proposition}{Proposition}
\newtheorem{remark}{Remark}
\newtheorem{corollary}{Corollary}
\theoremstyle{definition}
\newtheorem{definition}[theorem]{Definition}
\newtheorem*{example}{Example}

\algnewcommand{\Inputs}[1]{%
  \State \textbf{Inputs:}
  \Statex \hspace*{\algorithmicindent}\parbox[t]{.8\linewidth}{\raggedright #1}
}
\algnewcommand{\Initialize}[1]{%
  \State \textbf{Initialize:}
  \Statex \hspace*{\algorithmicindent}\parbox[t]{.8\linewidth}{\raggedright #1}
}

\def\tsc#1{\csdef{#1}{\textsc{\lowercase{#1}}\xspace}}
\tsc{WGM}
\tsc{QE}

\newenvironment{keywords}{%
  \par\medskip\begingroup\small\noindent\textbf{Keywords:}\enspace
}{\par\endgroup\medskip}
\newcommand{\sep}{\unskip;\ }

\makeatletter
\renewcommand{\@maketitle}{%
  \newpage\null\vskip 0.3em
  \begin{center}%
    {\LARGE\bfseries \@title\par}%
    \vskip 1.1em
    {\normalsize \@author\par}%
    \vskip 0.7em
    {\small
      Department of Mathematics and Computer Science,\\
      University of Cagliari, via Ospedale 72, Cagliari, 09124, Italy\par}%
    \vskip 0.5em
    {\footnotesize
      \href{mailto:b.cappellettimontano@unica.it}{\nolinkurl{b.cappellettimontano@unica.it}}\quad
      \href{mailto:mmusio@unica.it}{\nolinkurl{mmusio@unica.it}}\quad
      \href{mailto:nicola.piras97@unica.it}{\nolinkurl{nicola.piras97@unica.it}}\par}%
  \end{center}%
  \vskip 0.6em
}
\makeatother

\title{Regression models for ordinal compositional data}
\author{Beniamino Cappelletti-Montano,\quad Monica Musio,\quad
        Nicola Piras}
\date{}

\begin{document}
\maketitle

\begin{abstract}
We present a unified, transformation-free linear regression framework tailored for ordinal compositional data, such as distributions across educational levels or aggregate Likert-scale survey responses. Traditional log-ratio approaches often obscure interpretability, ignore the inherent ordering of categories, and struggle with boundary zero values. To overcome these limitations, we propose a deterministic model constrained to the space of column-stochastic transformation matrices, naturally preserving the fundamental geometry of the simplex. By adopting the weighted 1-Wasserstein distance as the loss function, our method directly embeds the ordinal nature of the data into the estimation process. We provide a computationally efficient and globally optimal solution explicitly formulated as a Linear Programming (LP) problem. The framework systematically addresses four distinct regression scenarios, introducing a novel ordinal tensor product to rigorously handle interactions across multiple compositional predictors. Furthermore, we equip the model with a suitable regularization strategy and novel diagnostic tools, including a Wasserstein-based coefficient of determination ($R^2_W$) and an Order Preservation Index (OPI). Extensive simulation studies and two empirical applications demonstrate the robust predictive performance and interpretability of the proposed methodology.
\end{abstract}


\section{Introduction}\label{sec1}

Compositional data, defined as vectors of non-negative proportions summing to a constant, are ubiquitous across scientific disciplines \citep{Aitchison1}. Typical examples include the mineralogical composition of rocks in geology, chemical concentrations in laboratory samples, or the allocation of household expenditures and time use in economics and social sciences. While regression models for compositional data have been extensively studied (see, e.g., \cite{Chen, Li}), they often rely on log-ratio transformations---such as the Additive Log-Ratio (ALR) or the Isometric Log-Ratio (ILR)---to map the constrained simplex into an unconstrained Euclidean space \citep{Aitchison2, Aitchison3}. Despite their popularity, these transformations present significant disadvantages: they can obscure the interpretability of model parameters, and they generally cannot handle zero values, which are frequent in empirical datasets. Moreover, the underlying Aitchison geometry is inherently nominal. It treats categories as exchangeable, making it unsuitable for ordinal compositional data where the relative position of categories carries essential information. For instance, in a Likert scale composition, the ``distance'' between \textit{Satisfied} and \textit{Very Satisfied} should be conceptually smaller than that between \textit{Unsatisfied} and \textit{Very Satisfied}. Standard log-ratio approaches fail to ``sense'' this proximity.

The need to model ordered compositions arises in various fields \citep{Weiss}. In social sciences, data often represent distributions across educational levels (e.g., high school graduation, bachelor's, master's, and PhD); in medicine and psychology, compositions frequently derive from questionnaires evaluating opinions or service quality on ordinal scales. These scenarios involve ordinal data \citep{Agresti}, where categories follow a natural progression. While traditional methods for analyzing ordinal responses are typically based on Generalized Linear Models (GLMs) and their extensions \citep{McCullagh2}, these models are primarily designed to handle discrete categorical outcomes at the individual level. Extending them to model continuous, aggregate compositional vectors while rigorously preserving the unit-sum constraint and the geometry of the simplex remains methodologically cumbersome. Motivated by the limitations of transformation-based methods and the growing interest in transformation-free models that operate directly on the simplex \citep{Fiksel}, we develop a unified regression framework for mappings of the form:
\begin{equation}\label{modello}
f:\Delta_{n_1}\times\cdots\times\Delta_{n_p}\longrightarrow    \Delta_{m_1}\times\cdots\times\Delta_{m_q}
\end{equation}
where a product of compositional predictors is mapped to a product of ordinal compositional outcomes. Our approach is distinguished by the use of the weighted 1-Wasserstein distance as a loss function. By treating each composition as a discrete probability distribution on an ordered set, the Wasserstein metric naturally incorporates the ordinal structure into the estimation process. To ensure structural consistency, the regression mappings are constrained to be column-stochastic matrices, naturally preserving the compositional nature of the output without requiring data transformations. Furthermore, by exploiting the $L_1$-nature of the Wasserstein distance, we provide a deterministic estimation strategy formulated exactly as a Linear Programming (LP) problem. This ensures computational tractability and guarantees global optimality.
Investigating the general model in \eqref{modello} requires addressing four fundamental regression scenarios:
\begin{enumerate}
    \item[(I)] a single compositional predictor to a single compositional outcome: $f: \Delta_n \longrightarrow \Delta_m$;
    \item[(II)] a single compositional predictor to a product of compositional outcomes: $f: \Delta_n \longrightarrow \Delta_m \times \Delta_k$;
    \item[(III)] a product of compositional predictors to a single compositional outcome: $f: \Delta_n \times \Delta_h \longrightarrow \Delta_m$;
    \item[(IV)] a product of compositional predictors to a product of compositional outcomes: $f: \Delta_n \times \Delta_h \longrightarrow \Delta_m \times \Delta_k$.
\end{enumerate}

To handle interaction terms across multiple input dimensions (Cases III and IV), we introduce a novel ordinal tensor product construction that embeds the joint space into a higher-dimensional simplex while strictly preserving the underlying metric geometry. In our framework, we assume the response variables to be ordinal, directly embedding this feature via the Wasserstein distance. Conversely, the model does not impose ordinal constraints on the explanatory variables, which may be ordinal or nominal, ensuring broad flexibility in empirical applications. By bypassing log-ratio transformations and leveraging the LP formulation, our proposed method remains perfectly robust in the presence of zeros and provides a highly interpretable geometric mapping between compositions. We also embed a natural regularization strategy to enforce smooth transitions between adjacent categories, and we propose novel goodness-of-fit ($R^2_W$) and order-preservation (\textrm{OPI}) indices explicitly tailored to the Wasserstein geometry.

The paper is structured as follows: Section \ref{sec2} defines the mathematical notation and the geometry of the Wasserstein order; Section \ref{sec3} presents the proposed regression framework, including the LP formulation, regularized estimation, and goodness-of-fit measures; Section \ref{sec4} reports the results of simulation studies assessing the predictive performance and robustness of the model; Section \ref{sec5} presents an application to European educational attainment data; Section \ref{sec6} analyzes a survey dataset on urban life satisfaction; finally, Section \ref{sec7} concludes the paper. Supplemental material and mathematical proofs are also provided in the Appendices.

\section{Ordinal structure and Wasserstein geometry on the simplex}\label{sec2}
In all the following sections, we denote the $n$-dimensional standard simplex by
\begin{equation*}
\Delta_{n}:=\left\{\left(x_{1}, \ldots, x_{n+1}\right)\in\mathbb{R}^{n+1} : x_i\geq 0 \textrm{ for all } i=1,\ldots,n+1, \textrm{ and }   x_{1} + \cdots + x_{n+1}=1 \right\}.
\end{equation*}
Usually, any element of $\Delta_n$ is called a \emph{composition}. In many situations, there is a natural ordering among the vertices of the simplex. For example, in a customer satisfaction survey with three categories—``satisfied'', ``neutral'', and ``dissatisfied''—each point of the 2-dimensional simplex represents a possible outcome of the survey. We can associate the vertex $V_1=(1,0,0)$ with the ``best result'', where 100\% of customers report being satisfied with the service; the vertex $V_2=(0,1,0)$ with the ``intermediate result'', where 100\% of customers report being neutral; and the vertex $V_3=(0,0,1)$ with the ``worst'' outcome, where all customers report being dissatisfied.  In these situations, whenever we are dealing with ordinal variables, we adopt the convention of ordering the vertices of the simplex associating the vertex $V_1=(1,0,\ldots,0)$ to the ``best'' outcome and so on until the vertex $V_{n+1}=(0,0,...,1)$ associated to the ``worst'' outcome.
The aim of this section is to extend this natural order to a total order relation defined for all the points of $\Delta_n$. We start by introducing a distance on $\Delta_n$, which we shall use through all the article. Let $\mathbf{a}=(a_{1}, \ldots, a_{n})$ be a $n$-uple of positive real numbers. The \emph{weighted $1$-Wasserstein distance} with weights $a_{1}, \ldots, a_{n}$, is  the function $W_{1, \mathbf{a}}: \Delta_n \times \Delta_n \longrightarrow [0, +\infty)$ defined by
\begin{equation}\label{wass}
W_{1, \mathbf{a}}(P,Q)=\sum_{i=1}^{n} a_{i}\left| F_{i}(P) - F_{i}(Q) \right| 
\end{equation}
where $F_{i}(P)=F(x_{1},\ldots,x_{n+1})=x_{1}+\cdots+x_{i}$ denotes the $i$-th component of the cumulative distribution function of $P\in\Delta_{n}$ (see \cite{CMontano2}). This definition stems from the Optimal Transport theory (see \cite{Villani}). Indeed  \eqref{wass} corresponds to the classical  $1$-Wasserstein distance when the underlying set of categories is equipped with the metric $d(i,j)=\sum_{k=i}^{j-1} a_k$,  for any $i <j$. From this point of view, the constants $a_i$ quantify the cost for moving
from a category to the one immediately higher. Their choice can be dictated by domain-specific knowledge, where the ``semantic gap'' between categories is non-uniform and the weights are calibrated to reflect the relative importance or the perceived ``effort'' associated with transitioning between states (see, e.g., \cite{CMontano1}, page 10). 
When such efforts are considered equivalent (i.e., the ground metric is the standard distance with $a_i=1$), we recover the classical discrete $1$-Wasserstein distance. 

In the context of ordinal compositional data, a well-known partial order relation is given by the \emph{first-order stochastic dominance} (FSD) \citep{Hadar}. Given $P,Q\in\Delta_n$, one says that $P$ \emph{stochastically dominates} $Q$ (in symbols, $Q  \preceq_{FSD} P$) if, by definition, $F_{i}(Q)\leq F_{i}(P)$ for each $i\in\left\{1,\ldots,n\right\}$. The dominance is \emph{strict}, and we write   $Q  \prec_{FSD} P$ if $F_{i}(Q)\leq F_{i}(P)$ for each $i\in\left\{1,\ldots,n\right\}$ and there exists $j\in\left\{1,\ldots,n\right\}$ such that $F_{j}(Q) < F_{j}(P)$.
The stochastic dominance is not a \emph{total} order relation, since there may be points  $P,Q \in \Delta_n$ such that $Q  \npreceq_{FSD} P$ and $P  \npreceq_{FSD} Q$. 
The $1$-Wasserstein distance provides a tool for defining a total order relation on $\Delta_n$, related to the first-order stochastic dominance, with respect to which any two points in $\Delta_n$ are comparable. Namely, we follow \cite{CMontano1}, where the authors associate to any $P=(x_1,\ldots,x_{n+1})\in\Delta_n$ an indicator which quantifies how far $P$ is from the best outcome $V_1$ by
\begin{equation}\label{delta_1}
\delta_1(P)=W_{1, \mathbf{a}}(P,V_1)= \sum_{j=1}^{n+1} c_j x_j
\end{equation}
where $c_1 = 0$ and $c_j = \sum_{k=1}^{j-1} a_k$ for $j > 1$  represents the cumulative distance of category $j$ from the best category. In particular, if $a_i = 1$ for all $i\in\left\{1,\ldots,n\right\}$, $\delta_1 (P) = \sum_{j=1}^{n+1} (j-1) x_j = \mathbb{E}[P] - 1$. 
We extend this construction by associating to a point $P\in\Delta_n$ a \emph{Wasserstein vector}
$\mathbf{W}(P)=\left(\delta_{1}(P), \delta_{2}(P), \ldots, \delta_{n+1}(P) \right)$,
where for each $i\in\left\{1,\ldots,n+1\right\}$, 
$\delta_{i}(P)=W_{1, \mathbf{a}}(P,V_i)$ is the distance from $P$ to the $i$-th vertex of $\Delta_n$.  Then we can define a new order relation on $\Delta_n$ as follows.
\begin{definition}
The \emph{Wasserstein order} is the relation $\preceq_{W}$ on $\Delta_n$ defined by
\begin{equation*}
  Q  \preceq_{W} P  \iff \mathbf{W}(P)\leq_{lex}\mathbf{W}(Q)
\end{equation*}
for any $P,Q\in\Delta_n$, where $\leq_{lex}$ denotes the lexicographical order.
\end{definition}
Thus, one first look at the distance of $P$ and $Q$ from $V_1$. If $\delta_1(P)<\delta_1(Q)$ then $Q\preceq_{W} P$. If  $\delta_{1}(P)=\delta_{1}(Q)$ then one examines the distance from $V_2$: if $\delta_{2}(P)<\delta_{2}(Q)$ then $Q\preceq_{W} P$. If also the distance from $V_2$ is equal, then one moves to the distance from $V_3$ and so on.

\begin{lemma}\label{ordine-totale}
The Wasserstein order $\preceq_{W}$ is a total order relation on $\Delta_n$. Specifically, for any $P, Q, Z \in \Delta_n$, the following properties hold:
\begin{enumerate}[(i)]
    \item \textit{Reflexivity:} $P \preceq_{W} P$;
    \item \textit{Transitivity:} If $P \preceq_{W} Q$ and $Q \preceq_{W} Z$, then $P \preceq_{W} Z$;
    \item \textit{Antisymmetry:} If $P \preceq_{W} Q$ and $Q \preceq_{W} P$, then $P = Q$;
    \item \textit{Totality:} either $P \preceq_{W} Q$ or $Q \preceq_{W} P$.
\end{enumerate}
\end{lemma}

\begin{proof}
The properties (i) and (ii) of $\preceq_{W}$ are fundamentally inherited from the properties of the lexicographical order $\le_{lex}$ on $\mathbb{R}^{n+1}$. \ As regards (iii), suppose $P \preceq_{W} Q$ and $Q \preceq_{W} P$. This implies $\mathbf{W}(Q) \le_{lex} \mathbf{W}(P)$ and $\mathbf{W}(P) \le_{lex} \mathbf{W}(Q)$. The antisymmetry of the lexicographical order yields $\mathbf{W}(P) = \mathbf{W}(Q)$, and consequently $\delta_i(P) = \delta_i(Q)$ for all $i \in \{1, \dots, n+1\}$.  Since a point in the simplex is uniquely determined by its distances to the vertices, it follows that $P = Q$. \ Finally, since  $\le_{lex}$ is a total order on $\mathbb{R}^{n+1}$, for any two vectors $\mathbf{W}(P)$ and $\mathbf{W}(Q)$, we must have either $\mathbf{W}(P) \le_{lex} \mathbf{W}(Q)$ or $\mathbf{W}(Q) \le_{lex} \mathbf{W}(P)$. According to the definition of $\preceq_{W}$, this directly translates to $Q \preceq_{W} P$ or $P \preceq_{W} Q$, respectively.
\end{proof}


The following proposition points out the link between Wasserstein order and FSD.

\begin{proposition}
Let $P$ and $Q$ be points of $\Delta_n$. If $P \preceq_{FSD} Q$, then $P \preceq_{W} Q$
\end{proposition}
\begin{proof}
By hypothesis, $P \preceq_{FSD} Q$ implies $F_i(P) \leq F_i(Q)$ for all $i\in\left\{1,\ldots,n\right\}$. 
Since the weights $a_i$ are strictly positive, the previous inequality implies, for each $i\in\left\{1,\ldots,n\right\}$, $a_i (1 - F_i(Q)) \le a_i (1 - F_i(P))$.
Summing over all indices $i = 1, \dots, n$, we obtain:
\begin{equation*}
\delta_1(Q)=W_{1, \mathbf{a}}(Q, V_1)=\sum_{i=1}^n a_i (1 - F_i(Q)) \le \sum_{i=1}^n a_i (1 - F_i(P))=W_{1, \mathbf{a}}(P, V_1)=\delta_1(P).
\end{equation*}
Hence $P \preceq_{W} Q$.
\end{proof}



In order to extend the Wasserstein order to the product of simplices, we recall the notion of tensorial product of two compositions. For any $P=(x_1, \ldots, x_{n+1}) \in \mathbb{R}^{n+1}$ and $Q=(y_1, \ldots, y_{h+1}) \in \mathbb{R}^{h+1}$ the \emph{tensorial (or Kronecker) product} of $P$ by $Q$ is defined by
\begin{equation}\label{tensor-product}
P \otimes Q = \left(x_{1}y_{1}, \ldots, x_{1}y_{h+1}, \ldots, x_{n+1}y_{1}, \ldots, x_{n+1}y_{h+1}\right) \in \mathbb{R}^{(n+1)(h+1)}.
\end{equation}
We have the following result.

\begin{lemma}\label{tensor}
If $P$ and $Q$ are compositions, then also $P \otimes Q$ is a composition.
\end{lemma}
\begin{proof}
Let $P=\left(x_{1}, \ldots, x_{n+1}\right)\in\Delta_{n}$ and $Q=\left(y_{1}, \ldots, y_{h+1}\right)\in\Delta_{h}$. Then
\begin{equation*}
P \otimes Q =\left(x_{1}, \ldots, x_{n+1}\right) \otimes \left(y_{1}, \ldots, y_{h+1}\right)= \left(x_{1}y_{1}, \ldots, x_{1}y_{h+1}, \ldots, x_{n+1}y_{1}, \ldots, x_{n+1}y_{h+1}\right)
\end{equation*}
and, using that $\sum_{i=1}^{n}x_i =1$ and  $\sum_{j=1}^hy_j =1$,  we have $x_{1}y_{1}+\cdots+ x_{1}y_{h+1} +\cdots+ x_{n+1}y_{1}+ \cdots+x_{n+1}y_{h+1}=x_{1}(y_{1}+\cdots+y_{h+1})+\cdots+ x_{n+1}(y_{1}+\cdots+y_{h+1})
=(x_{1}+\cdots+x_{n+1})(y_{1}+\cdots+y_{h+1})=1\cdot 1 =1$.
Thus $P \otimes Q \in \Delta_{nh+n+h}$.
\end{proof}

In the case where $P=(x_1, \ldots, x_{n+1}) \in \Delta_n$ and $Q =(y_{1}, \ldots, y_{h+1})\in \Delta_h$ represent \textit{ordinal} compositions, we use a different ordering of the components $x_i y_j$ in \eqref{tensor-product}, in order to take into account the ordinal nature  of the compositions. Let $\mathcal{I}=\left\{1,\ldots,n+1\right\} \times \left\{1,\ldots,h+1\right\}$ be the set of joint states. We consider a \textit{cost map} 
\begin{equation}\label{mappa-costo}
    \mathcal{C}: \mathcal{I}\longrightarrow\mathbb{R}
\end{equation}
and a bijection $\sigma: \{1, \dots, (n+1)(h+1)\} \longrightarrow \mathcal{I}$ such that $\mathcal{C}(\sigma(k)) \le \mathcal{C}(\sigma(k+1))$ for all $k\in\left\{1,\ldots,(n+1)(k+1)\right\}$. The choice of the map \eqref{mappa-costo} and of the correspondent reordering map $\sigma$ may depend by several factors, such as the underlying geometry or the specific context. A natural choice, for example, is to take $\mathcal{C}(i,j):=i+j$ and to define
\begin{equation}\label{ordinal-tensor-product}
P \otimes Q =\left(x_{1}, \ldots, x_{n+1}\right) \otimes \left(y_{1}, \ldots, y_{h+1}\right)= \left(x_{1}y_{1}, x_{1}y_{2}, x_{2}y_{1}, \ldots, x_{n+1}y_{h+1}\right)
\end{equation}
However, to ensure that the joint construction induces a total order and preserves the metric information of the marginals, the choice of the map \eqref{mappa-costo} should be guided  by the underlying geometry of the simplices $\Delta_n$ and $\Delta_h$. 
Suppose $\Delta_n$ and $\Delta_h$ are endowed with weighted 1-Wasserstein distances with weights $\mathbf{v}=(v_1, \ldots, v_n)$ and $\mathbf{u}=(u_1, \dots, u_h)$, respectively. Let $c_i = \sum_{k=1}^{i-1} v_k$ and $c'_j = \sum_{m=1}^{j-1} u_m$ be the cumulative costs, where $c_1=c'_1=0$. We define the \emph{ordinal cost function} as:
\begin{equation}\label{refined-cost}
\mathcal{C}(i, j) = c_i + c'_j + i\delta + j\eta
\end{equation}
where $\delta, \eta > 0$ are infinitesimal constants chosen to ensure that $\mathcal{C}$ is strictly injective on $\mathcal{I}$ (e.g., by taking $\delta$ and $\eta$ linearly independent over $\mathbb{Q}$). This ensures the uniqueness of the bijection $\sigma$, which now satisfies the strict inequality $\mathcal{C}(\sigma(k)) < \mathcal{C}(\sigma(k+1))$.  The  joint simplex $\Delta_{nh+n+h}$ is endowed with the weighted 1-Wasserstein distance with weights $\mathbf{w}=(w_1, \dots, w_{nh+n+h})$ defined as
$w_k = \mathcal{C}(\sigma(k+1)) - \mathcal{C}(\sigma(k)) > 0$.


\begin{remark}
The infinitesimal terms $\delta$ and $\eta$ in \eqref{refined-cost} serve to give a choice in the case of a tie, favoring the simplex $\Delta_n$ (taking $\delta$  much greater than $\eta$) or $\Delta_h$ (taking $\eta$  much greater than $\delta$), according to the user preferences. Moreover, they ensure that none of the weights $w_1, \ldots, w_{nh+n+h}$ of the 1-Wasserstein distance $W_{1, \mathbf{w}}$ of $\Delta_{nh+n+h}$ is zero (if that had happened, $W_{1, \mathbf{w}}$ would not have satisfied the identity of indiscernibles).
\end{remark}

\begin{example}\label{ex}
Let $\Delta_2$ ($n+1=3$) represent a primary quality dimension (e.g., product functionality) and $\Delta_3$ ($h+1=4$) represent a secondary dimension (e.g., delivery speed). Let us suppose that the weights of the weighted $1$-Wasserstein distance are $v_1 = 5, v_2 = 15$ for $\Delta_2$ and $u_1 = 2, u_2 = 4, u_3 = 8$ for $\Delta_3$. To ensure a strictly total order, we apply the ordinal cost function $\mathcal{C}(i, j) = c_i + c'_j + i\delta + j\eta$. Table \ref{esempio} shows the components of these refined costs.
\begin{table}[!ht]
\centering
\caption{Ordinal costs $\mathcal{C}(i, j)$ and corresponding ranking via $\sigma$}\label{esempio}
\begin{tabular}{@{}ccccc@{}}
\toprule
$(i, j)$ & $c_i = \sum a$ & $c'_j = \sum b$ & {Ordinal cost $\mathcal{C}(i, j)$} & {Rank ($\sigma$)} \\ \midrule
(1, 1) & 0  & 0  & $0 + \delta + \eta$    & 1 \\
(1, 2) & 0  & 2  & $2 + \delta + 2\eta$   & 2 \\
(2, 1) & 5  & 0  & $5 + 2\delta + \eta$   & 3 \\
(1, 3) & 0  & 6  & $6 + \delta + 3\eta$   & 4 \\
(2, 2) & 5  & 2  & $7 + 2\delta + 2\eta$  & 5 \\
(2, 3) & 5  & 6  & $11 + 2\delta + 3\eta$ & 6 \\
(1, 4) & 0  & 14 & $14 + \delta + 4\eta$  & 7 \\
(2, 4) & 5  & 14 & $19 + 2\delta + 4\eta$ & 8 \\
(3, 1) & 20 & 0  & $20 + 3\delta + \eta$  & 9 \\
(3, 2) & 20 & 2  & $22 + 3\delta + 2\eta$ & 10 \\
(3, 3) & 20 & 6  & $26 + 3\delta + 3\eta$ & 11 \\
(3, 4) & 20 & 14 & $34 + 3\delta + 4\eta$ & 12 \\ \bottomrule
\end{tabular}
\end{table}
In this specific case, for any $P=(x_{1},x_{2},x_{3})\in\Delta_2$ and $Q=(y_{1},y_{2},y_{3},y_{4})\in\Delta_3$, the ordinal tensor product is defined as:
\begin{align*}
P\otimes Q &= \left(x_{1}y_{1}, x_{1}y_{2}, x_{2}y_{1}, x_{1}y_{3}, x_{2}y_{2}, x_{2}y_{3}, x_{1}y_{4}, x_{2}y_{4}, x_{3}y_{1}, x_{3}y_{2}, x_{3}y_{3}, x_{3}y_{4}\right).
\end{align*}
The joint weight vector $\mathbf{w} = (w_1, \dots, w_{11})$ in $\Delta_{11}$ is obtained by the differences between successive refined costs. By selecting $\delta$ and $\eta$ such that $0 < \eta \ll \delta \ll 1$, we ensure that these infinitesimal terms remain subordinate to the primary cost increments induced by $\mathbf{a}$ and $\mathbf{b}$. Furthermore, this specific hierarchy assigns a priority to the first simplex in the event of metric ties. Indeed, if two states $(i,j)$ and $(s,t)$ satisfy $c_i+c'_j=c_s+c'_t$, the condition $\eta \ll \delta$ ensures that the refined costs are ordered according to the first index. 
Consequently, all weights $w_k$ are strictly positive and  are given by:
\begin{align*}
    w_1 &= \mathcal{C}(1,2) - \mathcal{C}(1,1) = 2 + \eta, \quad\quad\quad      w_2 = \mathcal{C}(2,1) - \mathcal{C}(1,2) = 3 + \delta - \eta,\\
    w_3 &= \mathcal{C}(1,3) - \mathcal{C}(2,1) = 1 - \delta + 2\eta, \quad  w_4 = \mathcal{C}(2,2) - \mathcal{C}(1,3) = 1 + \delta - \eta,\\
    w_5 &= \mathcal{C}(2,3) - \mathcal{C}(2,2) = 4 + \eta, \quad\quad\quad  w_6 = \mathcal{C}(1,4) - \mathcal{C}(2,3) = 3 - \delta + \eta,\\
    w_7 &= \mathcal{C}(2,4) - \mathcal{C}(1,4) = 5 + \delta, \quad\quad\quad  w_8 = \mathcal{C}(3,1) - \mathcal{C}(2,4) = 1 + \delta - 3\eta, \\
    w_9 &= \mathcal{C}(3,2) - \mathcal{C}(3,1) = 2 + \eta, \quad\quad\quad  w_{10} = \mathcal{C}(3,3) - \mathcal{C}(3,2) = 4 + \eta, \\
    w_{11} &= \mathcal{C}(3,4) - \mathcal{C}(3,3) = 8 + \eta.
\end{align*}
\end{example}
We now extend the Wasserstein order to the product of two simplices. A key role in this extension is played by the following lemma.

\begin{lemma}\label{segre-injective}
The map  $i: \Delta_n \times \Delta_h \longrightarrow \Delta_{nh+n+h}$  defined by $i(P, Q) = P \otimes Q$ is injective, i.e. for any $P,\tilde{P} \in\Delta_n$,  $Q,\tilde{Q}\in\Delta_h$, if $P \otimes Q = \tilde{P} \otimes \tilde{Q}$ then $P=\tilde{P}$ and $Q=\tilde{Q}$.
\end{lemma}
\begin{proof}
Let  $P = (x_1, \dots, x_{n+1}),  \tilde{P} = (\tilde{x}_1, \dots, \tilde{x}_{n+1}) \in \Delta_n$  and  $Q = (y_1, \dots, y_{h+1}),  \tilde{Q} = (\tilde{y}_1, \dots, \tilde{y}_{h+1}) \in \Delta_h$.
Suppose that  we have $i(P,Q)=i(\tilde{P},\tilde{Q})$, that is $P \otimes Q = \tilde{P} \otimes \tilde{Q}$.  This equality   implies that for every pair of indices $i \in \{1, \dots, n+1\}$ and $j \in \{1, \dots, h+1\}$, the corresponding components are equal:
\begin{equation}\label{prodotto-componenti}
x_i y_j = \tilde{x}_i \tilde{y}_j
\end{equation}
Let us fix the index $i$ and sum both sides of \eqref{prodotto-componenti} over all possible values of $j$, yielding $\sum_{j=1}^{h+1} (x_i y_j) = \sum_{j=1}^{h+1} (\tilde{x}_i \tilde{y}_j)$. Using that $x_i$ and $\tilde{x}_i$ do not depend on the summation index $j$ and  $\sum_{j=1}^{h+1} y_j =\sum_{j=1}^{h+1} \tilde{y}_j = 1$, as in $Q, \tilde{Q} \in \Delta_{h}$, the previous equality becomes
$x_i=x_{i}\cdot 1=x_i  \sum_{j=1}^{h+1} y_j  = \tilde{x}_i  \sum_{j=1}^{h+1} \tilde{y}_j =\tilde{x}_{i}\cdot 1 = \tilde{x}_i$.
Since this equality holds for every $i \in \{1, \dots, n+1\}$, it follows that $P = \tilde{P}$. In a similar way one proves  that $Q = \tilde{Q}$.
\end{proof}

The map defined in Lemma \ref{segre-injective} allows us to apply the Wasserstein order construction to the image $i(P,Q)$, for any $P\in\Delta_n, Q\in\Delta_h$. We then associate each point $X=P\otimes Q\in\Delta_{nh+n+h}$  with a distance profile vector $\mathbf{W}(X) = \left(\delta_{1}(X), \delta_{2}(X), \dots, \delta_{(n+1)(h+1)}(X)\right)$, with respect to the weighed $1$-Wasserstein distance $W_{1, \mathbf{w}}$ on $\Delta_{nh+n+h}$. The total order relation $\preceq_{W}$ on $\Delta_n \times \Delta_h$ is formally defined by the lexicographic comparison of these distance profiles:
\begin{equation*}
(P, Q) \preceq_{W} (\tilde{P}, \tilde{Q}) \iff \mathbf{W}(\tilde{P} \otimes \tilde{Q})   \le_{lex} 
\mathbf{W}(P \otimes Q) 
\end{equation*}
where $\le_{lex}$ is the standard lexicographic order on $\mathbb{R}^{(n+1)(h+1)}$. This construction satisfies all the axioms of a total order relation. In particular, reflexivity, transitivity and totality follow directly from the properties $\le_{lex}$, as in Lemma \ref{ordine-totale}. The only non-trivial property to be checked is the antisymmetry: 
 if $(P, Q) \preceq_{W} (\tilde{P}, \tilde{Q})$ and $(\tilde{P}, \tilde{Q}) \preceq_{W} (P, Q)$, then $\mathbf{W}(P \otimes Q) = \mathbf{W}(\tilde{P} \otimes \tilde{Q})$. Since a point in a simplex is uniquely determined by the set of its $1$-Wasserstein distances from the vertices, it follows that $P \otimes Q = \tilde{P} \otimes \tilde{Q}$. Furthermore, by Lemma \ref{segre-injective}, the equality $P \otimes Q = \tilde{P} \otimes \tilde{Q}$ implies $(P, Q) = (\tilde{P}, \tilde{Q})$.

In Appendix A we demonstrate the metric consistency of our construction by proving that the joint distance from the optimal vertex is additive with respect to the marginal factors.\\ 

The Wasserstein order may be used for providing a measure of the ordinal association between two sets of compositional ordinal data, including the case in which data belong to product of simplices. Let $\left\{(P_i, P'_i)\right\}_{i=1}^N$ be $N$ pairs of ordinal data, where $P_i \in \Delta_n$ and $P'_i \in \Delta_m$ for each $i\in\left\{1,\ldots,N\right\}$. We want to understand if there is association between $\left\{P_1, \ldots, P_N\right\}$ and $\left\{P'_1, \ldots, P'_N\right\}$.  To address this, we can define a non-parametric measure of association, the \textit{Ordinal Compositional Correlation} (\textrm{OCC}), as the Spearman's rank correlation between these two sets of data according to the rank induced by the Wasserstein order in each simplex $\Delta_n$ and $\Delta_m$. 
The \textrm{OCC} takes values in $[-1,1]$ and can be interpreted analogously to standard correlation measures. In particular, values close to $+1$ indicate a strong positive monotonic association, whereby predictor compositions with mass concentrated at higher indices tend to be associated with response compositions exhibiting similarly located mass. Conversely, values near $-1$ reflect a strong negative association, while values around $0$ suggest the absence of a monotonic relationship between the predictor’s center of mass and the ordinal level of the response.


\medskip

The $1$-Wasserstein distance can be extended also a distance in the set $\mathcal{CS}^{(m+1)\times(n+1)}$ of column-stochastic matrices, which will play a central role in our regression models. A matrix ${A} \in \mathbb{R}^{(m+1)\times(n+1)}$ is said to be \emph{column-stochastic} if its entries ${A}_{ij} \ge 0$ for all $i,j$, and $\sum_{i=1}^{m+1} {A}_{ij} = 1$ for all $j=1,\ldots,n+1$.
Geometrically, this implies that each column of ${A}$, denoted by ${A}_{\cdot j}$, is an element of the simplex $\Delta_m$.
Using this geometric property, we can naturally extend the weighted $1$-Wasserstein distance to the space of column-stochastic matrices. Given a weight vector $\mathbf{a}=(a_1, \dots, a_m)$ associated with the output space $\Delta_m$, the weighted $1$-Wasserstein distance between two matrices ${A}, {B} \in \mathcal{CS}^{(m+1)\times(n+1)}$ is defined as the sum of the distances between their corresponding columns:
\begin{equation}\label{wass_matrix}
W_{1, \mathbf{a}}({A}, {B}) := \sum_{j=1}^{n+1} W_{1, \mathbf{a}}({A}_{\cdot j}, {B}_{\cdot j}) = \sum_{j=1}^{n+1} \sum_{k=1}^{m} a_k \left| F_{k}({A}_{\cdot j}) - F_{k}({B}_{\cdot j}) \right|.
\end{equation}

\section{The Regression Framework}\label{sec3}
In this section, we present a unified framework for regression analysis where input and output variables belong to simplex spaces or, more in general, products of simplexes.
As mentioned in Section \ref{sec1}, the general model \eqref{modello} can be splitted into four cases. We begin with the core model mapping a simplex to another simplex (Case I) and then demonstrate how this prototype naturally extends to handle product spaces (Cases II, III, and IV).  

\subsection{Core Model and LP formulation}
We focus on the regression of an explanatory variable lying in the input simplex $\Delta_n$ to an ordinal response variable in the output simplex $\Delta_m$.
Consider a sample of $N$ pairs $\{({P}_i, {P}'_i)\}_{i=1}^N$. Here, ${P}_i \in \Delta_n$ denotes the predictor, while ${P}'_i \in \Delta_m$ denotes the ordinal response distribution. We assume a linear dependence structure defined by $f({P}) = {A}{P}$, where the regression coefficient ${A} \in \mathcal{CS}^{(m+1)\times(n+1)}$ is a column-stochastic matrix.
The estimation of ${A}$ is performed by minimizing the loss function $L({A})$, defined as the sum of the weighted $1$-Wasserstein distances in $\Delta_m$ between the predicted values ${A}{P}_i$ and the observed targets ${P}'_i$:
\begin{equation*}
    \mathcal{L}({A}) = \sum_{i=1}^N W_{1, \mathbf{a}}({A}{P}_i, {P}'_i) = \sum_{i=1}^N \sum_{k=1}^{m} a_k \left| F_{k}({A}{P}_i) - F_{k}({P}'_i) \right|.
\end{equation*}
Formally, our estimator $\hat{{A}}$ is defined as the matrix that minimizes the total error over the set of column-stochastic matrices:
\begin{equation}
    \hat{{A}} = \underset{{A} \in \mathcal{CS}^{(m+1)\times(n+1)}}{\arg\min} \sum_{i=1}^{N} W_{1, \mathbf{a}}({A}{P}_i, {P}'_i).
    \label{eq:estimator_def}
\end{equation}
Since the objective function involves $L_1$-norm terms, this optimization problem can be reformulated exactly as a Linear Program (LP) by introducing auxiliary variables $t_{ik}$ to handle the nonlinearity. The resulting formulation is:
\begin{equation*}
    \min_{{A}, \{t_{ik}\}} \quad \sum_{i=1}^N \sum_{k=1}^{m} a_k t_{ik}
\end{equation*}
subject to the structural constraints on ${A}$  and the linearization constraints for the absolute values:
\begin{equation}\label{LP-constraints}
    t_{ik} \ge F_{k}({A}{P}_i) - F_{k}({P}'_i), \quad 
    t_{ik} \ge -\left(F_{k}({A}{P}_i) - F_{k}({P}'_i)\right), \quad
    {A} \in \mathcal{CS}^{(m+1)\times(n+1)}, 
\end{equation}
for all $i\in\left\{1,\dots,N\right\}$ and $k\in\left\{1,\dots,m\right\}$.
This linear formulation ensures global optimality and can be efficiently solved using standard optimization software.

\subsection{Case II: Product Outcome }
This scenario models a mapping from a single predictor in $\Delta_n$ to a pair of ordinal response variables lying in $\Delta_m$ and $\Delta_k$, respectively.
Given a dataset $\{({P}_i, ({P}'_i, {Q}'_i))\}_{i\in\left\{1,\ldots, N\right\}}$, with ${P}'_i \in \Delta_m$ and ${Q}'_i \in \Delta_k$, the model is    $f({P}) = ({A}{P}, {B}{P})$, where ${A} \in \mathcal{CS}^{(m+1)\times(n+1)}$ and $B \in \mathcal{CS}^{(k+1)\times(n+1)}$. To estimate the parameters, we consider weight vectors $\mathbf{a}=(a_1,\ldots,a_m)$ for $\Delta_m$ and $\mathbf{b}=(b_1,\ldots,b_k)$ for $\Delta_k$. The loss function is the sum of the individual Wasserstein distances:
\begin{equation*}
    \mathcal{L}({A}, {B}) = \sum_{i=1}^N \left( W_{1, \mathbf{a}}({A}{P}_i, {P}'_i) + W_{1, \mathbf{b}}({B}{P}_i, {Q}'_i) \right).
\end{equation*}
Since the parameters ${A}$ and ${B}$ do not interact in the objective function, the optimization problem breaks down into two independent Case I regressions. We solve for $\hat{{A}}$ using $\{({P}_i, {P}'_i)\}$ and for $\hat{{B}}$ using $\{({P}_i, {Q}'_i)\}$, utilizing the LP formulation in Section 3.1.

\subsection{Case III: Product Input }\label{prodotto1}
Here, the explanatory variable is a pair $({P}, {Q})$ with ${P} \in \Delta_n$ and ${Q} \in \Delta_h$. To capture the interaction between inputs, we employ the tensor product. For each observation $i$, we construct the vector    ${Z}_i = {P}_i \otimes {Q}_i \in \mathbb{R}^{(n+1)(h+1)}$.
By Lemma \ref{tensor}, ${Z}_i$ is a valid composition residing in the higher-dimensional simplex $\Delta_{D}$, where $D = nh + n + h$.
The regression model is
    $f({P}, {Q}) = {A}({P} \otimes {Q})$,
where ${A} \in \mathcal{CS}^{(m+1)\times(D+1)}$.
This case effectively reduces to Case I. By transforming each input pair into a single vector ${Z}_i$, we treat the problem as a standard regression from $\Delta_{D}$ to $\Delta_m$, solvable via the minimization of $\sum_i W_{1, \mathbf{a}}({A}{Z}_i, {P}'_i)$.

\subsection{Case IV: Product Input and Product Outcome}\label{prodotto2}
Finally, we consider the general case mapping $\Delta_n \times \Delta_h \longrightarrow \Delta_m \times \Delta_k$. The model combines the strategies from the previous scenarios:
 $f({P}, {Q}) = \left({A}({P} \otimes {Q}), \; {B}({P} \otimes {Q})\right)$.
Operationally, we first pre-process the inputs to construct ${Z}_i = {P}_i \otimes {Q}_i \in \Delta_D$. Subsequently, the problem separates into two independent tasks: solving for ${A}$ using $\{({Z}_i, {P}'_i)\}$ and for ${B}$ using $\{({Z}_i, {Q}'_i)\}$. This demonstrates the modularity of our framework, as even complex product-space architectures are resolved by iteratively applying the canonical linear programming estimator.

\subsection{Regularized estimation}\label{regularization}

In the regression scenarios previously described, particularly when the dimensionality of the input space increases (as in Cases III and IV) or when the sample size $N$ is small, the estimation of the transformation matrix ${A}$ may become unstable or subject to overfitting. In such contexts, the standard minimization of the loss function (15) might yield a solution where adjacent columns of ${A}$ exhibit sharp variations. This is often undesirable because, in many applications, one expects the response distribution to change smoothly as the input moves along its ordinal categories. To address this, we introduce a regularization framework specifically designed to respect the ordinal geometry of the simplex. The core idea is to add two penalty terms based on the 1-Wasserstein distance between columns of ${A}$ that correspond to adjacent categories in the input space.

\subsubsection{Regularization for the Base Case (Case I)} 
In the core model mapping $\Delta_n$ to $\Delta_m$, the regularized objective function is:
\begin{align*}
\mathcal{L}_{reg}({A})=\sum_{i=1}^N W_{1, \mathbf{a}}({A}{P}_i, {P}'_i)
+ &\lambda_1 \sum_{j=1}^n \sum_{k=1}^{m} a_k \left| F_{k,j+1} - F_{k,j} \right|\\
 +& \lambda_2 \sum_{j=2}^{n} \sum_{k=1}^{m} a_k \left| F_{k,j-1} - 2F_{k,j} + F_{k,j+1} \right|
\end{align*}
where, for each $k\in \left\{1,\dots,m\right\}$, $F_{k,j} = \sum_{r=1}^k A_{rj}$ is the cumulative representation associated with the $j$-th column $A_{\cdot j}$, and $\lambda_1, \lambda_2 \geq 0$ are tuning parameters controlling the trade-off between goodness-of-fit and smoothness of the estimated conditional distributions. 
While these penalty terms are algebraically expressed via absolute differences to preserve the linearity of the optimization problem, they can be interpreted   within   Wasserstein geometry. Namely, the first-order penalty translates precisely to the sum of the distances between adjacent conditional distributions:
\begin{equation*}
    \sum_{j=1}^n \sum_{k=1}^{m} a_k \left| F_{k,j+1} - F_{k,j} \right| = \sum_{j=1}^n W_{1, \mathbf{a}}\left(A_{\cdot j+1}, A_{\cdot j}\right).
\end{equation*}
Moreover, by factoring out a constant $2$, the inner term of the second-order variation can be algebraically rewritten as   $ |F_{k,j-1} - 2F_{k,j} + F_{k,j+1}| = 2 \left| \frac{F_{k,j-1} + F_{k,j+1}}{2} - F_{k,j} \right|$.
Because the mapping from a probability mass function to its cumulative distribution is linear, the quantity $(F_{k,j-1} + F_{k,j+1})/2$ is exactly the cumulative distribution function of the statistical mixture  between the two adjacent columns, namely $\frac{1}{2}A_{\cdot j-1} + \frac{1}{2}A_{\cdot j+1}$. Since any convex combination of points in a simplex remains within the simplex, this mixture is a point of $\Delta_m$. Consequently, the second-order regularization term corresponds to
\begin{equation*}
    2 \sum_{j=2}^{n} \sum_{k=1}^{m} a_k \left| \frac{F_{k,j-1} + F_{k,j+1}}{2} - F_{k,j} \right| = 2 \sum_{j=2}^{n} W_{1, \mathbf{a}}\left( \frac{A_{\cdot j-1} + A_{\cdot j+1}}{2}, A_{\cdot j} \right).
\end{equation*}
Thus  the $\lambda_2$ penalty  evaluates a 1-Wasserstein distance, forcing the $j$-th conditional distribution to be as close as possible to the barycentric mixture of its neighbors.

\subsubsection{Extension to Product Input Spaces (Cases III and IV)}
In Cases III and IV, the explanatory variable is represented by the vector ${P} \otimes {Q}$, which resides in the higher-dimensional simplex $\Delta_{nh+n+h}$. As detailed in Section \ref{sec2}, the components of the tensor product are uniquely arranged according to an ordinal cost function $\mathcal{C}$ and permutation $\sigma$. This construction, which maps the multidimensional interaction of the predictors onto a single, unified ordinal length scale $(n+1)(h+1)$, ensures that two columns $j$ and $l$ of the  matrix ${A}$ are  adjacent if they correspond to consecutive states in the total order established by the mapping  $\sigma$. Thus, as in the previous case, we can define the set of adjacent indices as $\mathcal{E} = \left\{ (j, l) : |j - l| = 1 \right\}$, where $j, l \in \{1, \dots, (n+1)(h+1)\}$. By adopting this linear definition of adjacency, the regularization penalty promotes smoothness along the path of increasing joint costs. This approach ensures that as the input ``moves'' from one joint state to the next according to the Wasserstein total order, the predicted response distribution changes progressively, effectively leveraging the integrated ordinal structure of the product space.

\subsubsection{Extension to Product Outcomes (Cases II and IV)}
When the response consists of multiple compositions (e.g., ${A}$ and ${B}$), the total regularized loss is the sum of the individual regularized objectives. Since the parameters do not interact, we can optimize each matrix independently, by $\mathcal{L}_{reg}({A}, {B}) = \mathcal{L}_{reg}({A}) + \mathcal{L}_{reg}({B})$, where each term follows the logic of the single-outcome scenario with its respective tuning parameter ($\lambda_A, \lambda_B$).

\subsubsection{Regularized LP Formulation}
Since the weighted 1-Wasserstein distance $W_{1, \mathbf{a}}$ is defined as an $L_1$ norm of differences between cumulative distribution functions, the regularization preserves the linearity of the optimization problem. 
%
The regularized problem is solved via the following Linear Program:
\begin{align*}
\min_{{A},\{t_{ik}\},\{u_{jk}\},\{v_{jk}\}}
 \sum_{i=1}^N \sum_{k=1}^m a_k t_{ik}
+ \lambda_1 \sum_{j=1}^n \sum_{k=1}^m a_k u_{jk}
+ \lambda_2 \sum_{j=2}^n \sum_{k=1}^m a_k v_{jk}
\end{align*}
subject to the  constraints \eqref{LP-constraints}  for $t_{ik}$ and ${A}$, and the additional constraints associated with the first- and second-order regularization terms. For the first-order regularization:
\begin{equation*}
u_{jk} \geq F_{k,j} - F_{k,j+1},    \quad
u_{jk} \geq -(F_{k,j} - F_{k,j+1}),
\end{equation*}
for all $j=1,\ldots ,n$ and $k=1,\ldots,m$, while for the second-order regularization
\begin{equation*}
v_{jk}
\geq
F_{k,j-1} - 2F_{k,j} + F_{k,j+1}, \quad
v_{jk}
\geq
-\left(F_{k,j-1} - 2F_{k,j} + F_{k,j+1}\right),
\end{equation*}
fo rall $j=2,\dots,n$ and $ k=1,\dots,m$.
The first-order regularization promotes similarity between adjacent conditional distributions and tends to induce piecewise-constant structures across neighboring columns. This is particularly useful when abrupt changes occur. The second-order regularization instead penalizes variations in the discrete gradient across adjacent columns, enforcing smoother transitions in the ordered index space.
Depending on the application, one may employ either the first-order term, the second-order term, or both simultaneously. In practice, the first-order penalty is more suitable when the underlying structure is expected to exhibit segmentation (for example, due to a very low value of $N$), whereas the second-order penalty is preferable when smooth functional evolution across columns is assumed, this term becomes particularly appropriate when the number of columns is large. For further details regarding the selection procedure of the regularization parameters see Appendix B.

\subsection{Measuring Goodness-of-Fit: the Wasserstein $R^2$ and the Order Preservation Index}
\label{ssec:goodness_of_fit}

After fitting the model $f({P}) = {A}{P}$ (or its extensions to product spaces), its performance must be assessed. To this end, we introduce two complementary indices: a coefficient of determination, $R^2_W$, capturing metric accuracy, and an Order Preservation Index (\textrm{OPI}), designed to evaluate the structural monotonicity of the regression.
We begin by defining a baseline model that ignores the predictors and systematically predicts the  Fréchet mean $\bar{{P}}' \in \Delta_m$ of the points $P'_{1}, \ldots, P'_{N}$, defined by
\begin{equation*}
\bar{{P}}' = \underset{P' \in \Delta_m}{\arg\min} \sum_{i=1}^N W_{1,\mathbf{a}}(P', {P}'_i).
\end{equation*}
As proved in Appendix E, $\bar{{P}}'$ is the point of $\Delta_m$ whose CDF components are given by
$\bar{F}_k = \mathrm{median}\left \{ F_k({P}'_1), \ldots, F_k({P}'_N)   \right\}$,
for each $k\in\left\{1,\ldots,m\right\}$. The prediction error of this baseline captures the overall ordinal dispersion in the data and is defined as the \emph{Total Wasserstein Dispersion}:
\begin{equation}
{\mathcal{D}}_{\textrm{tot}} = \sum_{i=1}^N W_{1, \mathbf{a}}(\bar{{P}}', {P}'_i).
\label{eq:SST_W}
\end{equation}
In the case of product outcomes (Cases II and IV), since the loss function is defined additively across the marginal spaces, ${\mathcal{D}}_{\textrm{tot}}$ naturally decomposes as the sum of the marginal total Wasserstein dispersions. We then quantify the variability unexplained by the fitted model through the sum of Wasserstein distances between the predicted and observed compositions,
that we call \emph{Residual Wasserstein Dispersion}
\begin{equation}
{\mathcal{D}}_{\textrm{res}} = \sum_{i=1}^N W_{1, \mathbf{a}}({f({P}_i)}, {{P}'_i}).
\label{eq:SSE_W}
\end{equation}
The \emph{Wasserstein coefficient of determination} $R^2_W$ is defined as the proportional reduction in error relative to the baseline:
\begin{equation}
R^2_W = 1-\frac{{\mathcal{D}}_{\textrm{res}}}{{\mathcal{D}}_{\textrm{tot}}}
\label{eq:R2_W}
\end{equation}
While $R^2_W$ represents a conceptual analog to the classical Euclidean $R^2$, the Wasserstein geometry does not naturally yield an additive decomposition of variance. Nevertheless, thanks to our specific LP formulation and the choice of the true Fréchet mean as the baseline, $R^2_W$ is mathematically guaranteed to be strictly bounded between $0$ and $1$ in-sample, making it a robust diagnostic tool.

\begin{proposition}\label{prop:r2_bounds}
For any fitted model obtained via the unregularized or regularized Linear Programming framework, the in-sample Wasserstein coefficient of determination strictly satisfies $0 \le R^2_W \le 1$, assuming ${\mathcal{D}}_{\textrm{tot}}>0$.
\end{proposition}
\begin{proof}
We define a ``null model'' transformation matrix ${A}_{\mathrm{base}} \in \mathbb{R}^{(m+1)\times(n+1)}$ where every single column is exactly equal to Fréchet mean $\bar{{P}}'$. Since $\bar{{P}}' \in \Delta_m$, ${A}_{\mathrm{base}}$ is a column-stochastic matrix. Consequently, ${A}_{\mathrm{base}}$ is a mathematically feasible solution within our optimization space.  Applying ${A}_{\mathrm{base}}$ to any predictor yields $({A}_{\mathrm{base}} {P}_i)_{j} = \sum_{k=1}^{n+1} (\bar{{P}}^\prime)_j ({P}_i)_k = (\bar{{P}}')_j$, since ${P}_i\in\Delta_n$. Thus, ${A}_{\mathrm{base}} {P}_i=\bar{P}'$ for any $i\in\left\{1,\ldots,N\right\}$, so that the unregularized empirical loss evaluated at ${A}_{\mathrm{base}}$ is exactly the total dispersion:
\begin{equation*}
\mathcal{L}({A}_{\mathrm{base}}) = \sum_{i=1}^N W_{1,\mathbf{a}}({A}_{\mathrm{base}} {P}_i, {P}'_i) = \sum_{i=1}^N W_{1,\mathbf{a}}(\bar{{P}}', {P}'_i) = \mathcal{D}_{\textrm{tot}}.    
\end{equation*}
The unregularized estimator $\hat{{A}}$ minimizes the empirical loss $\mathcal{L}({A})$ over $\mathcal{CS}^{(m+1)\times(n+1)}$. Since ${A}_{\mathrm{base}} \in \mathcal{CS}^{(m+1)\times(n+1)}$, global optimality guarantees that
$\mathcal{D}_{\textrm{res}} = \mathcal{L}(\hat{{A}}) \le \mathcal{L}({A}_{\mathrm{base}}) = \mathcal{D}_{\textrm{tot}}$.
This property extends to the regularized estimator $\hat{{A}}_{\mathrm{reg}}$, which minimizes $\mathcal{L}_{\mathrm{reg}}({A}) = \mathcal{L}({A}) + \mathrm{Pen}({A})$, where $\mathrm{Pen}({A})$  denotes the sum of the non-negative first and second-order regularization penalty terms introduced in Section \ref{regularization}:
\begin{equation*}
\mathrm{Pen}({A})=\lambda_1 \sum_{j=1}^n \sum_{k=1}^{m} a_k \left| F_{k,j+1} - F_{k,j} \right|
+ \lambda_2 \sum_{j=2}^{n} \sum_{k=1}^{m} a_k \left| F_{k,j-1} - 2F_{k,j} + F_{k,j+1} \right|.
\end{equation*}
Because all columns in ${A}_{\mathrm{base}}$ are identical, the cumulative distribution of each column is the same ($F_{k,j} = \bar{F}_k$). Therefore, the adjacent variations evaluated by the first and second-order penalty terms are identically zero, yielding $\mathrm{Pen}({A}_{\mathrm{base}}) = 0$. Hence, $\mathcal{L}_{\mathrm{reg}}({A}_{\mathrm{base}}) = \mathcal{L}({A}_{\mathrm{base}}) + 0 = \mathcal{D}_{\textrm{tot}}$. By the global optimality of $\hat{{A}}_{\mathrm{reg}}$, we get $\mathcal{L}_{\mathrm{reg}}(\hat{{A}}_{\mathrm{reg}}) \le \mathcal{L}_{\mathrm{reg}}({A}_{\mathrm{base}})$, which yields
$\mathcal{D}_{\textrm{res}}\leq \mathcal{D}_{\textrm{res}} + \mathrm{Pen}(\hat{{A}}_{\mathrm{reg}}) \le \mathcal{D}_{\textrm{tot}}$.
Dividing the inequality by $\mathcal{D}_{\textrm{tot}} > 0$ yields $0 \le \frac{\mathcal{D}_{\textrm{res}}}{\mathcal{D}_{\textrm{tot}}} \le 1$, which immediately gives $0 \le R^2_W \le 1$.
\end{proof}

The choice of the median as a baseline implies that $R^2_W$ must be interpreted as a relative improvement over a highly robust central tendency. Consequently, it is inherently more conservative than its Euclidean counterpart based on the mean; moderate values of $R^2_W$ (e.g., $0.3-0.5$) already indicate a substantial explanatory power.

While $R^2_W$ quantifies metric proximity in the Wasserstein geometry, it does not explicitly capture the preservation of the ordinal structure. Thus, when all variables are ordinal, we define an indicator evaluating the monotonicity of the mapping.
Let $\prec$ be any pre-order relation on the simplex (or product of simplices). Let $\textrm{sgn}({P}, {Q}) \in \{1, -1, 0\}$ be the sign of the ordinal relationship between two points according to $\prec$, defined as $\textrm{sgn}({P},{Q})=1$ if ${P}\prec {Q}$, $\textrm{sgn}({P},{Q})=-1$ if ${P}\succ {Q}$, and $0$ otherwise. 
The \emph{Order Preservation Index} (\textrm{OPI}) with respect to the partial order relation $\prec$ is defined as the proportion of data pairs whose ordinal rank is preserved by the model, evaluated among all pairs with distinct observed responses:
\begin{equation}
\textrm{OPI}_{\prec} = \frac{\left| \left\{ (i, j)\in\left\{1,\ldots,N\right\}^2 : i < j \text{ and } \textrm{sgn}({P}'_i, {P}'_j) \cdot \textrm{sgn}(f({P}_i), f({P}_j)) = 1 \right\} \right|}{\left| \left\{ (i, j)\in\left\{1,\ldots,N\right\}^2  : i < j \text{ and } {P}'_i \neq {P}'_j \right\} \right|}
\label{eq:OPI}
\end{equation}
In our context, a natural choice is to evaluate the index using the total Wasserstein order $\preceq_{W}$. The two proposed measures act in a complementary way: $R^2_W$ assesses whether the predictions are structurally \emph{close} to the target distributions, while $\textrm{OPI}_{\preceq_{W}}$ confirms whether the model correctly preserves the \emph{ordinal direction} of the variations, an essential requirement for ordinal regression tasks.

\subsection{Inference via Bootstrap}
To construct confidence regions for the estimated matrix $\hat{{A}}$, we rely on a non-parametric bootstrap procedure \citep{Efron}. We give the detailed constriction for the Case I, the other cases being similar. 
The procedure involves generating $R$ bootstrap datasets by resampling $N$ pairs with replacement from the original data. For each bootstrap sample $r\in\left\{1, \dots, R\right\}$, we re-estimate the matrix ${A}^*_r$ using the LP solver. This yields an empirical distribution of matrices $\{{A}^*_1, \ldots, {A}^*_R\}$.
We compute the distance between the original estimate $\hat{{A}}$ and each bootstrap estimate ${A}^*_r$ using the  $1$-Wasserstein distance for column-stochastic matrices (see Eq. \eqref{wass_matrix}):
\begin{equation}
     W_{1, \mathbf{a}}({A}^*_r, \hat{{A}}) = \sum_{j=1}^{n+1} W_{1, \mathbf{a}}(({A}^*_r)_{\cdot j}, \hat{{A}}_{\cdot j}) = \sum_{j=1}^{n+1} \sum_{k=1}^{m} a_k \left| F_{k}(({A}^*_r)_{\cdot j}) - F_{k}(\hat{{A}}_{\cdot j}) \right|.
\end{equation}
We define the radius of the $95\%$ confidence region, denoted by $r_{0.95}$, as the $95$-th percentile of the empirical distribution of distances $\{W_{1, \mathbf{a}}({A}^*_k, \hat{{A}})\}_{k=1}^R$. 
To facilitate interpretation across different model dimensions, we also compute a normalized radius, denoted by $\tilde{r}_{0.95}$. This is obtained by dividing $r_{0.95}$ by the theoretical diameter of the space, which represents the maximum possible distance between any two column-stochastic matrices:
\begin{equation}
\tilde{r}_{0.95} = \dfrac{r_{0.95}}{(n+1)\sum_{k=1}^m a_k}.
\end{equation}
Finally, a summary estimator for the confidence region can be derived by computing the Fréchet mean of the bootstrap estimates that lie within the closed ball of radius $r_{0.95}$
\begin{equation*}
\overline{B}\left(\hat{A}, r_{0.95}\right) =\left\{{A}^*_k \in \mathcal{CS}^{(m+1)\times(n+1)} : W_{1, \mathbf{a}}({A}^*_k, \hat{{A}}) \le r_{0.95}\right\}.
\end{equation*}

\section{Simulation studies}\label{sec4}

In this section, we investigate the performance of the proposed methodology under two distinct simulation settings, each designed to highlight different aspects of the approach.
In the first setting, the main objective is to benchmark our method against the CODALM R package \citep{codalm} from a predictive standpoint. To this end, we employ two proper scoring rules, namely the Ranked Probability Score (RPS) and the Brier Score (BS). Beyond simulation-based comparisons, we further assess the empirical behavior and stability of the competing methods through a K-fold cross-validation procedure applied to the Educ dataset (see Section \ref{sec5}). This dataset, also used in CODALM, is particularly relevant as it involves both compositional ordinal predictors and responses.
In the second setting, the comparison with CODALM is no longer pursued. Instead, the focus shifts to a deeper and more controlled evaluation of the proposed method itself. Specifically, we concentrate on Cases I and III, as Cases II and IV can be obtained as direct extensions of these two fundamental scenarios. The performance is assessed using metrics that are intrinsic to our framework. Namely, we shall use the \emph{normalized Residual Wasserstein Dispersion}, defined as
$n\mathcal{D}_{\textrm{res}}=\frac{\mathcal{D}_{\textrm{res}}}{N\left(a_{1}+\cdots + a_{m}\right)}$,  which captures the discrepancy between predicted and observed compositions, together with the $\textrm{OPI}_{\preceq_W}$ and $R^2_W$ indices, which respectively quantify the preservation of the ordinal structure and the proportion of variability explained by the model.

The  R code to generate data and perform simulations, is made available on a Github repository at the link \url{https://github.com/NicolaPiras97/OrdinalCompositionalData}. Furthermore, the proposed methodology is implemented in the \texttt{OrdinalCompositions} \citep{package} R package, which is publicly available for reproducible application of the proposed approach.

 \subsection{Ordinal Data Generating Processes}\label{gen}

The central element of the simulation design is the construction of
conditional distributions $P' \mid P$ that embed explicit ordinal structure. For each observation $i$, we define a continuous latent ordinal score $s_i = \sum_{k=1}^{n+1} k \, x_{ik}$. This score represents the expected ordinal location encoded by the distribution $P_i$. Since $P_i$ lies on the simplex, $s_i \in [1, n+1]$ and can be interpreted as the center of mass of $P_i$ along the ordered support. The conditional distribution of $P'_i$ is then constructed as a function of the distance between category index $j$ and the latent score $s_i$. Thus, ordinal structure is introduced by explicitly penalizing metric displacement on the support of $P'$. To generate sampling variability, once the target probability vector
$p_i = (p_{i1}, \dots, p_{im+1})$ is constructed,
we draw $P'_i \sim \text{Dirichlet}(\alpha \, p_i)$, where $\alpha > 0$ controls noise intensity. The conditional probability mass function is defined as:
\[
p_{ij} =
\frac{\exp\left(-\nu |j - s_i|\right)}
{\sum_{\ell=1}^{m+1} \exp\left(-\nu |\ell - s_i|\right)}.
\]
The distribution is unimodal and centered at $s_i$. The shape is symmetric around the latent score and $\nu$ controls concentration: with $\nu \to 0$ the distribution approaches uniformity.






\subsection{Simulation study: Setting 1}

To ensure a  robust comparison between the predictive performances of the two methodologies, we adopt the Ranked Probability Score (RPS)  (see \cite{Gneiting}) as a common evaluation metric on the validation set. The RPS is a strictly proper scoring rule specifically designed for probabilistic forecasts of ordinal variables. For any  $P=(x_{1},\ldots,x_{n+1}), Q=(y_{1},\ldots,y_{n+1}) \in \Delta_n$, it is defined as:
\begin{equation*}
\text{RPS}(P, Q) = \sum_{k=1}^{n} \left( F_k(P) - F_k(Q) \right)^2.
\end{equation*}
Together with the RPS we also consider the  Brier (BS) (see \cite{Brier}) defined as
\[
\text{BS}(P, Q)
=
\sum_{k=1}^{n+1}
(x_k - y_k)^2.
\]
The Brier score is strictly proper but treats categories as nominal, ignoring their ordering. We compare the performance of our model based on Wasserstein distance (WD) and the CODALM (COD) estimator in terms of RPS and BS.

Let $P \in \Delta_n$ and $P' \in \Delta_m$. Each Monte Carlo iteration consists of the following steps.

\subsubsection*{Step 1: Generation of the Initial Sample}
We generate $N$ independent observations as
$P_i \sim \text{Dirichlet}(1, \dots, 1)$,  $i = 1,\dots,N$. Conditional on each $P_i$, the response distribution $P'_i$ is generated according to the specified data-generating process (DGP).
Based on the sample $\{(P_i, P'_i)\}_{i\in\{1,\ldots,N\}}$, we estimate the matrix $A^{WD}$ using the simplex linear programming approach, and the matrix $A^{COD}$ using the COD algorithm.

\subsubsection*{Step 2: Generation of a New Sample}
An independent sample $\{P_i^{\text{new}}\}_{i=1}^N$ is drawn from the same Dirichlet distribution. Two alternative outcome samples are then generated as
\[
P_i^{'WD} \sim \text{Dirichlet}(\alpha \, A^{WD} P_i^{\text{new}}),
\quad
P_i^{'COD} \sim \text{Dirichlet}(\alpha \, A^{COD} P_i^{\text{new}}).
\]
This procedure induces two distinct evaluation settings, each corresponding to a specific underlying data-generating process (DGP):
\begin{itemize}
    \item[-] WD-DGP, if the true structure is WD-based;
    \item[-] COD-DGP, if the true structure is COD-based.
\end{itemize}

\subsubsection*{Step 3: Train/Test Split and evaluation Metrics}

The new sample is randomly split into 70\% training and 30\% validation. Both estimators are re-estimated on the training subsample and evaluated
out-of-sample on the validation set using RPS and BS. This procedure is described in detail in Algorithm 1 in Appendix C.

\medskip
Table~\ref{tab:winrates_N100} reports the empirical win rates of the WD and COD estimators under different data-generating processes and for various combinations of predictor and outcome categories $(n+1, m+1)$. Results are obtained via Monte Carlo simulations with sample sizes $N=100$ (Table~\ref{tab:winrates_N100}); additional results for $N=500$ are reported in Appendix D.
For each scenario, we report three win rates (see Appendix C for details):
\begin{itemize}
    \item[-] WD-DGP: proportion of replications in which WD outperforms COD when the data are generated under the WD-based DGP;
    \item[-] COD-DGP: proportion of replications in which WD outperforms COD when the data are generated under the COD-based DGP;
    \item[-] TRUE-DGP: cross-comparison, i.e., WD evaluated under its own DGP versus COD evaluated under its own DGP.
\end{itemize}
Increasing the sample size from $N=100$ to $N=500$ generally leads to improved win rates for the WD estimator, suggesting enhanced stability and identifiability of the WD methodology in larger samples. Under its own DGP, WD consistently achieves higher win rates than COD. When evaluated under the COD-DGP, the win rates of WD are typically close to $50\%$, indicating that WD retains a reasonable degree of robustness even when COD is the true operator.
RPS and BS exhibit qualitatively similar patterns. WD tends to dominate more clearly under RPS, particularly when the outcome support is larger; however, the relative ranking of the estimators remains broadly consistent across the two scoring rules. 

\begin{table}[!ht]
\centering
\caption{Win rates of WD method under different DGP ($N=100$).}\label{tab:winrates_N100}
\begin{tabular}{ccccccccc}
\toprule
 &  \multicolumn{2}{c}{WD-DGP} 
 & \multicolumn{2}{c}{COD-DGP}
 & \multicolumn{2}{c}{TRUE-DGP} \\
\cmidrule(lr){2-3}\cmidrule(lr){4-5}\cmidrule(lr){6-7}
$(n+1,m+1)$  & RPS & BS & RPS & BS & RPS & BS \\
\midrule

(3,3) 
        & 0.530 & 0.522 & 0.430 & 0.428 & 0.504 & 0.506 \\

\midrule

(7,7) 
        & 0.742 & 0.772 & 0.654 & 0.714 & 0.744 & 0.726 \\

\midrule

(5,5) 
       & 0.579 & 0.583 & 0.480 & 0.480 & 0.574 & 0.572 \\

\midrule

(3,5) 
       & 0.538 & 0.546 & 0.526 & 0.524 & 0.598 & 0.568 \\

\midrule

(5,3) 
        & 0.558 & 0.576 & 0.382 & 0.412 & 0.592 & 0.580 \\

\midrule

(5,7) 
       & 0.596 & 0.614 & 0.446 & 0.462 & 0.626 & 0.612 \\

\midrule

(7,5) 
       & 0.730 & 0.782 & 0.616 & 0.626 & 0.734 & 0.718 \\

\bottomrule
\end{tabular}
\end{table}

\subsection{Robustness to Sparsity and Anomalous Data}
A central element of the comparison between the Wasserstein-based (WD) and Kullback-Leibler-based (COD) approaches is their respective robustness to anomalous data and structural outliers. We designed a simulation study comprising three distinct structural perturbations, each intended to probe the limits of the estimators' robustness.
Throughout all scenarios, predictive performance is assessed using the Ranked Probability Score (RPS) and the Brier Score (BS). To summarise the comparative results, we report \textit{win rates}, defined as the proportion of Monte Carlo replications in which the WD-based method attains a lower validation error than COD.

\subsubsection{Scenario 1: unstructured sparsity}
To simulate baseline measurement error or accidental drop-out, we introduce unstructured sparsity in the validation set. From a theoretical standpoint, the superior robustness of the WD-based method originates in its loss function. In fact, minimising the 1-Wasserstein distance over ordered categories is geometrically equivalent to estimating a generalised median under an \(L_1\)-type norm on the cumulative distributions. In contrast, the COD approach minimises the Kullback-Leibler divergence, thereby acting as a Maximum Likelihood estimator of the conditional mean. 
We fix a proportion \(p \in [0, 0.2]\) of zero components relative to the total number of components. Let \(N\) denote the number of units in the validation set and \(m+1\) the dimension of each compositional vector. For a given \(p\), a total of \(p \cdot N \cdot (m+1)\) zeros are introduced. Units are randomly drawn so that the overall number of zeros matches the target. Within each selected unit, one component is set to zero when \(m+1=3\), while two components are zeroed when \(m+1 \in \{5, 7\}\); the compositions are then renormalised. For instance, with \(N = 30\), \(m+1 = 7\), and \(p = 0.2\), we introduce 42 zeros by randomly selecting 21 units and setting two of their components to zero.
The evolution of the average RPS as the proportion of zeros grows (Figure~\ref{fig:zeros}) reveals a clear divergence (see also Table 1 in Appendix D, which illustrates the robustness of our method). The WD-based approach displays stable behaviour, with RPS values staying nearly constant. Conversely, the COD-based method exhibits a systematic increase in RPS with \(p\), confirming that the KL-based estimator is highly sensitive to sparsity and suffers a progressive degradation of predictive accuracy.

\begin{figure}[!ht]
\centering
        \includegraphics[width=0.8\linewidth]{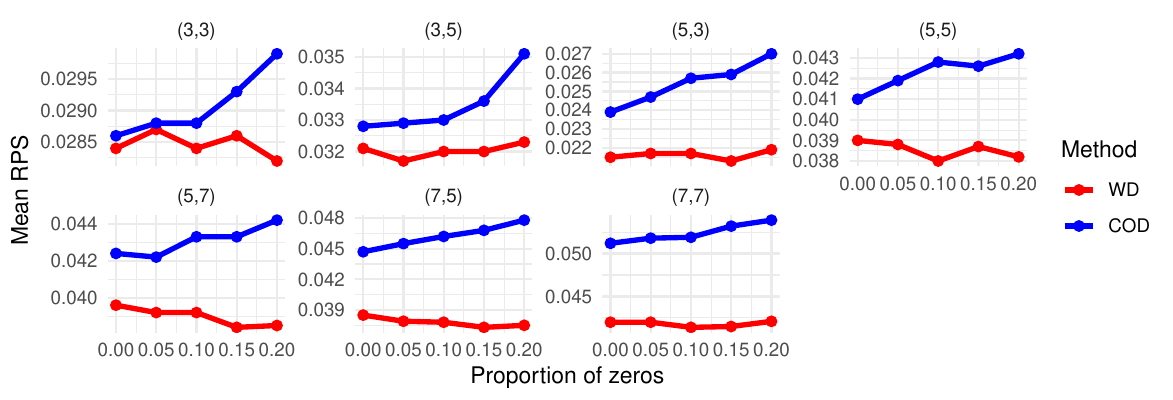}
       \caption{ Average RPS as a function of the proportion of zeros ($N=100$), evaluated under the TRUE-DGP. Each panel corresponds to a different $(n+1, m+1)$ configuration. Zeros are randomly selected.
} \label{fig:zeros}
       \end{figure}

\subsubsection{Scenario 2: mode-suppression stress test}
To evaluate the estimators' resilience under severe stress, we devise a worst-case adversarial contamination scheme. While random zero inflation acts as mild noise, explicitly forcing the \textit{largest} component (the mode) to zero aggressively destroys the primary correlation structure of the data. Methodologically, this constitutes a severe structural outlier that destroys the primary correlation structure of the data. From an applied perspective, the scenario models extreme anomalies such as ``protest votes'' or contrarian behaviour, where respondents deliberately avoid the expected consensus response.
Applying the same proportion logic (with \(p\) ranging from \(0\) to \(0.2\)), we now set to zero the largest component of each selected unit (or the two largest for \(m+1 \in \{5, 7\}\)), followed by renormalisation.
The results (see Table 2 in Appendix D) highlight several robust patterns. Under both the WD-DGP and the cross-DGP (TRUE-DGP) comparisons, the WD-based method consistently outperforms COD. Win rates are systematically above 0.5 and increase monotonically with the proportion of zeros, often exceeding 0.9. Figure~\ref{fig:zeros2} provides visual confirmation of the theoretical expectation: COD models tend to be sensitive to unexpected near-zero mass events, as the heavily penalised likelihood forces substantial adjustments in the regression coefficients and leads to a considerable increase in RPS. In contrast, the WD method remains stable, reflecting how its underlying \(L_1\)-geometry mimics a robust median estimator and naturally limits the influence of extreme structural outliers.

\begin{figure}[!ht]
\centering
        \includegraphics[width=0.8\linewidth]{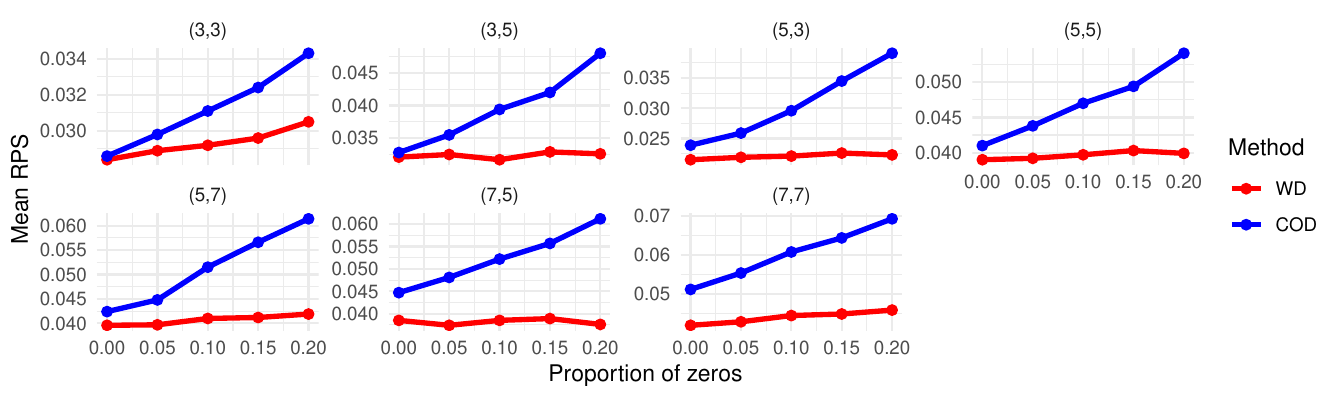}
       \caption{ Average RPS as a function of the proportion of zeros ($N=100$), evaluated under the TRUE-DGP. Each panel corresponds to a different $(n+1, m+1)$ configuration. The largest component in the selected data is fixed to zero.
} \label{fig:zeros2}
       \end{figure}

\subsubsection{Scenario 3: ordinal misunderstanding}
Finally, we evaluated the method's robustness to an outlier type highly specific to ordinal data, the scale inversion. For a given proportion of observations (\(p\) ranging from \(0\) to \(0.4\)) in the validation set, the order of the components in the response vector is completely reversed (e.g., \(P = (x_1, \dots, x_{m+1})\) becomes \(P' = (x_{m+1}, \dots, x_1)\)). This mimics situations where respondents misunderstand the direction of the ordinal scale, or where severe data-entry errors occur.
Because likelihood-based methods like COD treat categories as nominal and ignore the underlying metric space, they are highly sensitive to such symmetric probability shifts. As illustrated in Figure~\ref{fig:inv} and reported in Table 3 in Appendix D, the COD method suffers a substantial predictive failure as the inversion proportion increases. The WD-based method, by contrast, maintains strong and stable performance.

\begin{figure}[!ht]
\centering
        \includegraphics[width=0.8\linewidth]{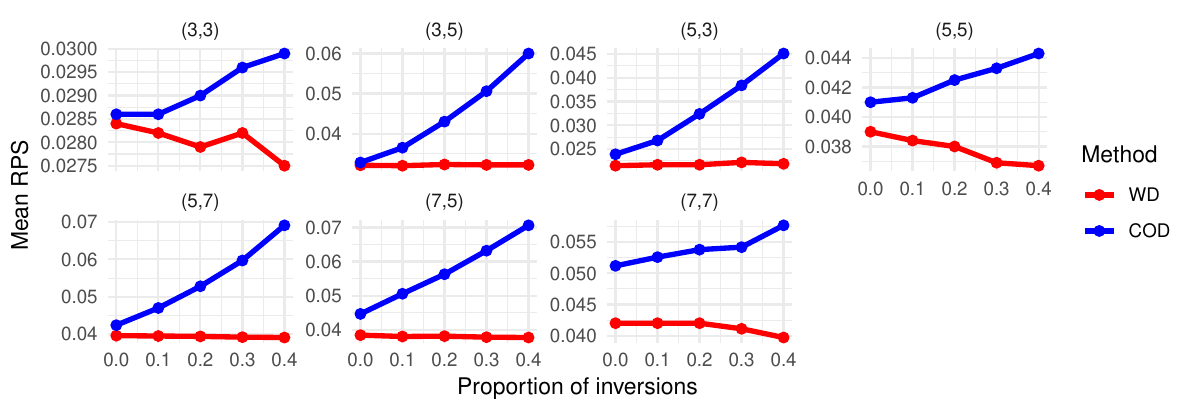}
       \caption{ Average RPS as a function of the proportion of inversions ($N=100$), evaluated under the TRUE-DGP. Each panel corresponds to a different $(n+1, m+1)$ configuration.
} \label{fig:inv}
       \end{figure}

\subsection{Simulation study: Setting 2}
We now evaluate the Wasserstein regression model independently of the COD estimator. Unlike the previous simulation---where samples were drawn under an operator enforcing linearity between predictors and outcomes---here we generate data directly from the target DGP and estimate the regression operator without imposing a linear mapping. 
The goal is to assess predictive accuracy under different weighting schemes, through metrics that capture full distributional accuracy ($n\mathcal{D}_{\textrm{res}}$), pairwise ordinal concordance (\(\textrm{OPI}_{\preceq_{W}}\)), and explained variation (\(R^2_W\)).

\subsubsection{Case I}

We first consider the core model with a single compositional covariate. Compositional ordinal data are generated following the procedure described in Section~\ref{gen}. From these data, we estimate the regression matrix by solving a linear programming problem that minimises the Wasserstein cost under the appropriate marginal constraints. Table~\ref{tab:w1_predictive} reports the values of $n\mathcal{D}_{\textrm{res}}$, together with \(\textrm{OPI}_{\preceq_{W}}\) and \(R^2_W\), across different dimensional configurations. The normalised  residual Wasserstein dispersion  is consistently small, ranging from 0.04 to 0.09 across all settings, indicating satisfactory predictive accuracy. The \(\textrm{OPI}_{\preceq_{W}}\) values are high overall (0.78--0.93), confirming that the model preserves pairwise ordinal concordance to a substantial degree, while \(R^2_W\) is moderate (0.30--0.45), suggesting that the covariate captures a meaningful, though not exhaustive, share of the distributional variation.
Performance tends to benefit from larger support sizes \((n+1, m+1)\), as expected given the richer latent structure. The highest \(\textrm{OPI}_{\preceq_{W}}\) (0.936) and \(R^2_W\) (0.452) are observed for the \((3,3)\) configuration, whereas the asymmetric setting \((5,3)\) records the lowest \(\textrm{OPI}_{\preceq_{W}}\) (0.779), a result likely attributable to the information loss incurred when projecting onto a lower-dimensional outcome space.
Results for \(N=500\), reported in Appendix D, confirm the stability of both \(R^2_W\) and $n\mathcal{D}_{\textrm{res}}$, and reveal a marked increase in the proportion of preserved orderings as the sample size grows.

\begin{table}[!ht]
\centering
\caption{Predictive performance of $n\mathcal{D}_{\textrm{res}}$, $\textrm{OPI}_{\preceq_{W}}$ and $R^2_W$, $N=100$, Setting 2, case I (standard dev. in parentheses) }\label{tab:w1_predictive}
\begin{tabular}{ccccc}
\toprule
$(n+1,m+1)$ & $R^2_W$ & $n\mathcal{D}_{\textrm{res}}$ & $\textrm{OPI}_{\preceq_{W}}$\\
\midrule

(3,3) 
 & 0.452 (0.039) & 0.074 (0.004) & 0.936 (0.010)  \\

\midrule

(7,7)
 & 0.297 (0.025) & 0.068 (0.004) & 0.918 (0.017)  \\

\midrule

(5,5)
 & 0.359 (0.033) & 0.072 (0.003) & 0.933 (0.015)  \\

\midrule

(3,5)
 & 0.446 (0.038) & 0.042 (0.002) &  0.934 (0.010)  \\

\midrule

(5,3)
 & 0.345 (0.051) & 0.068 (0.006) & 0.779 (0.026) \\

\midrule

(5,7)
 & 0.353 (0.033) & 0.049 (0.003) & 0.932 (0.015)  \\

\midrule

(7,5)
 & 0.290 (0.027) & 0.088 (0.005) & 0.910 (0.017) \\

\bottomrule
\end{tabular}
\end{table}
















\subsubsection{Case III: Two Compositional Covariates}
We extend the simulation design to assess the performance of the Wasserstein regression model in the presence of two compositional ordinal covariates.
Let \(P=(x_1,\dots ,x_{n+1})\in \Delta_n\), \(Q=(y_1,\dots ,y_{h+1}) \in \Delta_h\), and \(P'=(x'_1,\dots ,x'_{m+1}) \in \Delta_m\).
The regression model is estimated using the tensor representation \(Z = P \otimes Q\), solving the optimal transport regression problem under fixed ground weights. Both covariates are independently drawn from a uniform Dirichlet distribution: \(P \sim \text{Dir}(\mathbf{1})\) and \(Q \sim \text{Dir}(\mathbf{1})\).
For each observation \(i\), we compute the latent scores \(s_{1,i} = \sum_{k=1}^{n+1} k \, x_{ik}\) and \(s_{2,i} = \sum_{k=1}^{h+1} k \, y_{ik}\), and define their weighted combination \(s_i = w_1 s_{1,i} + w_2 s_{2,i}\), with \(w_1, w_2 > 0\) and \(w_1 + w_2 = 1\) (in the simulations presented below, \(w_1=0.6\) and \(w_2=0.4\)). Conditional on \(s_i\), the response distribution over the \(m+1\) ordinal categories is specified via a Laplace kernel by
$p_{ij} \propto \exp\!\left(-\nu |j - s_i|\right)$,  $j=1,\dots ,m+1$.
The final compositional response is then obtained as
$P'_i \sim \text{Dirichlet}(\alpha \, p_i)$,
where the concentration parameter \(\alpha\) governs the dispersion around the latent distribution.
Table~\ref{tab:w2_predictive} reports values of $n\mathcal{D}_{\textrm{res}}$, together with \(\textrm{OPI}_{\preceq_{W}}\) and \(R^2_W\), across different dimensional configurations. The values of $n\mathcal{D}_{\textrm{res}}$  remain consistently low, ranging from 0.02 to 0.05 across all settings, confirming the model's strong predictive accuracy also in the two-covariate case. The \(\textrm{OPI}_{\preceq_{W}}\) values are high throughout (0.89--0.92), indicating that ordinal concordance is well preserved even under the more complex tensor representation. At the same time, \(R^2_W\) reaches moderate-to-substantial levels (0.40--0.58), notably higher than those observed in the single-covariate setting, which suggests that the inclusion of a second compositional predictor yields a meaningful gain in explained distributional variation.
Configurations with symmetric or nearly symmetric dimensions, such as \((4,4,4)\) and \((3,5,7)\), attain the highest \(R^2_W\) (around 0.58), while the largest support size \((7,7,7)\) records a slightly lower \(R^2_W\) (0.402), plausibly reflecting the greater difficulty in recovering the full distributional structure when the outcome space is richer. The asymmetric configuration \((7,3,5)\) produces the lowest \(\textrm{OPI}_{\preceq_{W}}\) (0.896) and the highest $n\mathcal{D}_{\textrm{res}}$ (0.053), a pattern consistent with the information loss induced by compressing one covariate onto a low-dimensional support.
Results for \(N=500\), reported in Appendix D, confirm the overall stability of both \(R^2_W\) and $n\mathcal{D}_{\textrm{res}}$, and show a substantial increase in the proportion of preserved orderings as the sample size grows.

\begin{table}[!ht]
\centering
\caption{Predictive performance of $n\mathcal{D}_{\textrm{res}}$, $\textrm{OPI}_{\preceq_{W}}$ and $R^2_W$, $N=100$, Setting 2, case III (standard dev. in parentheses)}\label{tab:w2_predictive}
\begin{tabular}{ccccc}
\toprule
$(n+1,h+1,m+1)$ & $R^2_W$ & $n\mathcal{D}_{\textrm{res}}$  & $\textrm{OPI}_{\preceq_{W}}$  \\
\midrule

(5,5,5) 
 & 0.477 (0.029) & 0.042 (0.002) &  0.915 (0.012)  \\

\midrule

(4,4,4)
& 0.581 (0.023) & 0.043 (0.002) & 0.920 (0.009) \\

\midrule

(7,7,7)
& 0.402 (0.021) & 0.041 (0.002) & 0.913 (0.013) \\

\midrule

(3,5,7)
& 0.578 (0.024) & 0.021 (0.001) & 0.914 (0.009) \\

\midrule

(7,3,5)
& 0.417 (0.027) & 0.053 (0.003) & 0.896 (0.016)\\

\bottomrule
\end{tabular}
\end{table}












\section{Application: Educational attainment of parents in Europe}
\label{sec5}

We apply our metho-dology to a real-world dataset concerning the educational attainment of parents. We use the \texttt{educFM} dataset, available in the \texttt{robCompositions} R package, which contains the compositional breakdown of education levels for mothers and fathers across 31 European countries.
The data consists of three ordered categories: `high' (tertiary education), `medium' (upper secondary and post-secondary non-tertiary) and `low' (pre-primary, primary, or lower secondary). Our goal is to model the educational composition of fathers (${P}'_i \in \Delta_2$) as a linear function of the educational composition of mothers (${P}_i \in \Delta_2$). Given the clear ordinal nature of the response variable, this problem is perfectly suited for our Wasserstein-based approach.

Before fitting the regression model, we first quantify the underlying association in the raw data using the Ordinal Compositional Correlation. 
For the \texttt{educFM} dataset, we find $\text{OCC} = 0.8516$.
This strong positive correlation indicates that countries where mothers have a higher average educational attainment also tend to be the countries where fathers have a higher average educational attainment. 
We applied our regression model to estimate the transformation matrix $A$. The resulting point estimate $\hat{{A}}$, obtained with a selected regularization parameter $\lambda_1=0.1$ and $\lambda_2$ fixed to 0, is:
\begin{equation}
    \hat{{A}} = 
    \begin{bmatrix}
        0.8814 & 0.0561 & 0.0387  \\
        0.1186 & 0.9439 & 0.0282  \\
        0.0000 & 0.0000 & 0.9331 
    \end{bmatrix}
    \label{eq:educ_matrix}
\end{equation}
where we have used unitary weights for the 1-Wasserstein distance. The rows and columns of the matrix correspond to the categories (`high', `medium', `low') for fathers and mothers, respectively. The interpretation of this matrix is direct and intuitive. Each column represents the predicted educational composition of fathers given a 100\% concentration of mothers in that category. For instance, the third column implies that if the mothers in a country all have `low' educational attainment, the model predicts the fathers' composition would be (3.87\% `high', 2.82\% `medium', 93.31\% `low'). The matrix is strongly diagonal dominant, which is the key finding. This indicates a powerful positive association, a phenomenon often termed \emph{assortative mating} where individuals tend to partner with those of similar educational status.

To assess the variability of the estimated regression operator, a non-parametric bootstrap procedure with $R=1000$ resamples was performed. For each bootstrap replicate, the regression matrix was re-estimated and its Wasserstein distance from the original estimate was computed. Retaining only the matrices within the 95\% radius (where we have used the 1-Wasserstein distance among column-stochastic matrices given by \eqref{wass_matrix}) yields a subset of admissible solutions, whose columns are represented as compositional points in the simplex. The 95\% percentile confidence region has a normalized radius of 0.4202 and the Fréchet mean is
\begin{equation*}
    \hat{{A}} = 
    \begin{bmatrix}
        0.8055 & 0.1107 & 0.0350 \\
        0.1184 & 0.8832 & 0.0413 \\
        0.0761 & 0.0061 & 0.9237
    \end{bmatrix}
\end{equation*}
Figure~\ref{educplot} displays the resulting bootstrap confidence clouds. For graphical convenience, the simplices corresponding to the three columns (High, Medium, and Low) are jointly represented in a single triangular diagram, despite being conceptually distinct. Each cloud corresponds to one column of the regression matrix and represents the variability of the associated compositional effect. The dispersion of the points reflects the uncertainty of the estimates, while their relative location in the simplex highlights the dominant components. Overall, the concentration of the clouds suggests a stable estimation, with limited variability across bootstrap samples. A larger dispersion is observed for the first column (High), whose confidence cloud exhibits greater spread toward the top vertex, indicating increased uncertainty in the corresponding compositional effect.\\
\begin{figure}[!ht]
\centering
        \includegraphics[width=0.6\linewidth]{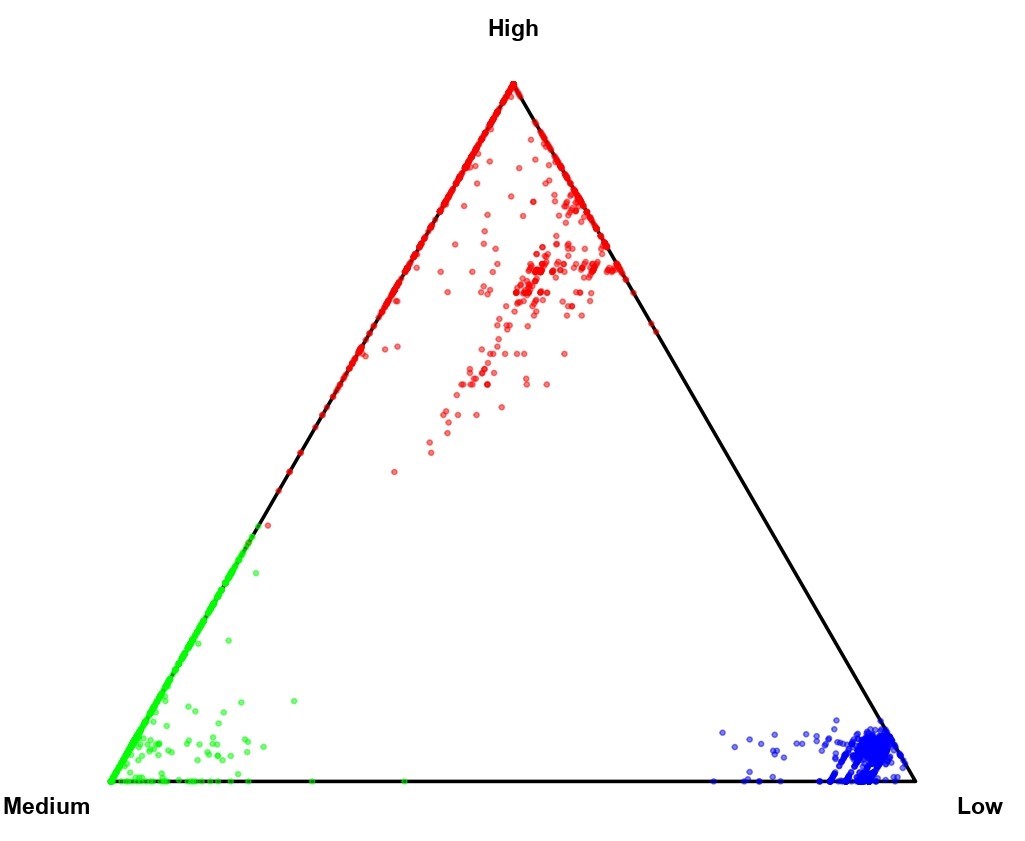}
       \caption{Bootstrap confidence clouds for the columns of the estimated regression matrix. Each point represents a bootstrap replicate falling within the 95\% Wasserstein radius from the original estimate. The simplex is displayed with vertices corresponding to High (top), Medium (left), and Low (right) categories. Colors identify the three columns of the regression operator, illustrating their variability and relative compositional structure.} \label{educplot}
       \end{figure}
To construct the $95\%$ bootstrap confidence intervals for each father's education level reported in Table~\ref{tab:conf}, we rely on the total Wasserstein order (see  Section \ref{sec1}). 
For each bootstrap replication $r = 1, \dots, R$, we obtain a bootstrap estimate $\hat{A}^{(r)}$. Each column in the collection $\{ \hat{A}^{(1)}, \dots, \hat{A}^{(R)} \}$ is then ordered according to the total Wasserstein order.
Using this ordering, we compute, the $2.5\%$ empirical quantile (lower bound), the Wasserstein median and the $97.5\%$ empirical quantile (upper bound).
The use of the total Wasserstein order ensures that bootstrap quantiles are computed in a way that is coherent with the geometry of the response space. Unlike component-wise intervals, which ignore dependence and ordering across categories, Wasserstein-based intervals preserve the cumulative structure of ordinal distributions.
The results show a strong diagonal dominance pattern: higher parental education levels are associated with higher probability mass concentrated in the upper response categories. This pattern is stable across the bootstrap distribution, as indicated by relatively tight confidence intervals in the dominant components.
Moreover, the intervals are asymmetric in several cases. This reflects the nonlinear geometry of the simplex under the $W_1$ metric and suggests that uncertainty is not uniformly distributed across categories. In particular, off-diagonal entries exhibit wider variability, indicating greater sensitivity to sampling fluctuations in transition probabilities between adjacent ordinal levels.

\begin{table}[ht]
\centering
\caption{95\% bootstrap confidence intervals for each mother's education level, R=10000.}\label{tab:conf}
\small
\resizebox{\columnwidth}{!}{%
\begin{tabular}{c|ccc|ccc|ccc}
\hline
& \multicolumn{3}{c|}{Mother's Educ. `low'} 
& \multicolumn{3}{c|}{Mother's Educ. `medium'}
& \multicolumn{3}{c}{Mother's Educ. `high'} \\
\hline
Father’s Educ. & Inf & Median & Sup 
& Inf & Median & Sup 
& Inf & Median & Sup \\
\hline
`low' 
& 1.0000 & 0.7013 & 0.3238 
& 0.3314 & 0.0761 & 0.0000 
& 0.0165 & 0.0525 & 0.0000 \\

`medium' 
& 0.0000 & 0.2987 & 0.5078 
& 0.6686 & 0.9144 & 0.9383 
& 0.1645 & 0.0066 & 0.0578 \\

`high' 
& 0.0000 & 0.0000 & 0.1684 
& 0.0000 & 0.0095 & 0.0617 
& 0.8189 & 0.9409 & 0.9422 \\
\hline
\end{tabular}
}
\end{table}
To evaluate the overall performance of the fitted model, we calculated the Wasserstein R-squared and the Order Preservation Index, obtaining
$R^2_{W} = 0.5296$ and $\textrm{OPI}_{\preceq_{W}}=0.8495$.
Finally we compare the proposed method with CODALM. We implement a repeated $K$-fold cross-validation procedure. Let $\{(P_i, P'_i)\}_{i=1}^N$ denote the observed sample, where $P_i \in \Delta_n$ represents the mother’s education distribution and $P'_i \in \Delta_m$ the father’s education distribution, both defined over ordered categories. The evaluation procedure is as follows:

\begin{enumerate}
\item The sample is randomly partitioned into $K=5$ folds of approximately equal size.
\item For each fold $k$:
\begin{itemize}
\item[-] the matrix $A$ is estimated using the data from the remaining $K-1$ folds (training set);
\item[-] out-of-sample predictions are computed for the held-out fold (test set) as $\hat{P'}_i = A P_i$.
\end{itemize}
\item The entire procedure is repeated multiple times (50) using different random partitions.
\end{enumerate}

This approach ensures that each observation is evaluated strictly out-of-sample,
providing an unbiased assessment of predictive performance. The results highlight a nuanced comparison between the proposed method and CODALM.
As shown in Table \ref{educ_comparison}, under the RPS our method exhibits a  similar  mean error, with a small positive mean difference. However, its win rate exceeds 52.52\%, suggesting that the two methods perform comparably when evaluated using a strictly proper ordinal scoring rule.
A similar pattern emerges under the Brier score. Although WD shows a marginally higher average error, it outperforms COD in more than $54.19\%$ of individual comparisons. Given that the Brier score does not account for the ordinal structure of the outcome, this result indicates that the proposed approach retains competitive probabilistic calibration even when the ordering information is disregarded.
\begin{table}[!ht]
\centering
\caption{Out-of-sample performance based on repeated 5-fold cross-validation.}\label{educ_comparison}
\begin{tabular}{lcccc}
\hline
Score & Method & Mean & Mean Difference & Win Rate \\
\hline
\multirow{2}{*}{RPS}  & WD & 0.0112 & \multirow{2}{*}{0.0009} & \multirow{2}{*}{52.52 \%} \\
      & COD   & 0.0103 &                        &       \\
\hline
\multirow{2}{*}{BS} & WD & 0.0167 & \multirow{2}{*}{0.0010} & \multirow{2}{*}{54.19 \%} \\
      & COD  & 0.0156 &                        &       \\
\hline
\end{tabular}
\end{table}

\section{Application to European quality of life data}
\label{sec6}
We apply our methodology to the Eurostat dataset \textit{``Quality of Life in European Cities''}, which surveys resident satisfaction across various urban dimensions (available at the website \url{https://ec.europa.eu/regional_policy/information-sources/maps/quality-of-life_en}). 
The dataset provides information on 83 European cities. Original responses were measured on a five-point Likert scale: \textit{Very satisfied} (VS), \textit{Rather satisfied} (RS), \textit{Rather unsatisfied} (RU), \textit{Very unsatisfied} (VU), and \textit{Don't know/No answer}. For this analysis, we excluded the \textit{Don't know/No answer} category due to its negligible frequency and re-normalized the remaining proportions to sum to unity.
The primary objective is to model the general satisfaction level of residents, measured by the survey item ``I am satisfied to live in my city'', which we treat as the response variable $Y$. For each city $i = 1, \dots, N$ (with $N=83$), the empirical distribution of responses to this item forms an ordinal composition denoted by ${P}'_i \in \Delta_3$. We aim to model the response composition ${P}'_i$ as a linear function of several covariates, representing specific urban services and lifestyle aspects. Among the available indicators, we initially selected the following survey variables to test their relationship with $Y$:
\begin{itemize}
    \item[-] ``Health care services, doctors and hospitals''
    \item[-] ``Green spaces such as parks and gardens''
    \item[-] ``The quality of the air''
    \item[-] ``It is easy to find a good job in my city''
    \item[-] ``I feel safe walking alone at night in my city''
    \item[-] ``It is easy to find good housing in my city at a reasonable price''.
\end{itemize}
We evaluated several regression models using different combinations of these covariates. The best-performing model, selected according to the $R_W^2$ metric, incorporates two specific explanatory variables, which we denote as $X_1$ and $X_2$. Specifically, the abstract variables and their corresponding compositional representations for the $i$-th city are defined as follows:
\begin{itemize}
    \item[-] $Y$: ``I am satisfied to live in my city'', with observed compositions ${P}'_i \in \Delta_3$;
    \item[-] $X_1$ (Covariate 1): ``I feel safe walking alone at night in my city'', with observed compositions ${P}_i \in \Delta_3$;
    \item[-] $X_2$ (Covariate 2): ``The quality of the air'', with observed compositions ${Q}_i \in \Delta_3$.
\end{itemize}
Before fitting the regression model, we assessed the monotonic association among these variables using the Ordinal Compositional Correlation (\textrm{OCC}) introduced in Section \ref{sec2}. Evaluated over the entire sample of 83 cities, the coefficients are: $\textrm{OCC}(X_1, Y) = 0.7326$,
     $\textrm{OCC}(X_2, Y) = 0.7420$,
      $\textrm{OCC}(X_1 \otimes X_2, Y) = 0.8284$.
Using our LP-based regression model with regularization parameters $\lambda_1=0$ $\lambda_2=0.1$, we estimated the transformation matrix $A$. The resulting point estimate, $\hat{A}$, is:
\[
\resizebox{\columnwidth}{!}{$
\hat{A} =
\left[
\begin{array}{cccccccccccccccc}
0.972 & 0.900 & 0.814 & 0.728 & 0.642 & 0.556 & 0.470 & 0.384 & 0.299 & 0.213 & 0.170 & 0.128 & 0.085 & 0.057 & 0.028 & 0.000 \\
0.000 & 0.005 & 0.041 & 0.131 & 0.282 & 0.433 & 0.524 & 0.616 & 0.701 & 0.787 & 0.830 & 0.770 & 0.588 & 0.392 & 0.196 & 0.000 \\
0.000 & 0.072 & 0.127 & 0.128 & 0.069 & 0.009 & 0.005 & 0.000 & 0.000 & 0.000 & 0.000 & 0.102 & 0.258 & 0.412 & 0.405 & 0.397 \\
0.028 & 0.023 & 0.018 & 0.013 & 0.007 & 0.002 & 0.001 & 0.000 & 0.000 & 0.000 & 0.000 & 0.000 & 0.069 & 0.139 & 0.371 & 0.603
\end{array}
\right]
$}
\]
where the weights used are $a_1=2, a_2=1, a_3=2$ for the response variable and the same weighting structure is adopted for both covariates in the tensor-product construction. The choice of the symmetric configuration $(2,1,2)$ reflects the ordinal structure of the categories while assigning greater importance to discrepancies occurring at the extreme levels of the distributions. In particular, the symmetric weighting configuration $(2,1,2)$ reflects the ordinal interpretation of the categories by assigning greater importance to transitions between extreme and intermediate states than to transitions between the two intermediate categories. In other words, moving from ``very satisfied'' to ``rather satisfied'', or from ``very dissatisfied'' to ``rather dissatisfied'', is considered more substantial than moving between ``rather satisfied'' and ``rather dissatisfied''. 
The symmetry of the weighting scheme guarantees that no directional preference is introduced between the positive and negative sides of the scale, while still emphasizing the stronger semantic separation associated with the extreme categories. 
The same weighting structure is adopted for both the response variable and the covariates in order to preserve coherence across all distributions involved in the analysis, since they are measured on analogous ordinal scales. This ensures that the Wasserstein geometry induced on both the predictor and response spaces reflects the same interpretation of category transitions and the same notion of ordinal discrepancy.
The rows of the matrix correspond to the categories (`VS', `RS', `RU', `VU') for the overall satisfaction while the columns correspond to the 16 combinations of modalities of covariates 1 (night-time safety) and 2 (air quality) described in Table \ref{comb}.
\begin{table}[!ht]\centering      %
\medskip
\caption{The 16 combinations of the categories of the two covariates $X_1$, $X_2$ in the order in which they appear in $\hat{{A}}$.}\label{comb}
\begin{tabular}{c|cc||c|cc}
\toprule[0.09 em]
{Column} & {$X_1$} & {$X_2$} & {Column} & {$X_1$} & {$X_2$}\\
\midrule
1 & VS & VS & 9 & RU & RS\\ 
2 & VS & RS & 10 & VU & VS\\ 
3 & RS & VS & 11 & RU & RU\\ 
4 & VS & RU & 12 & RS & VU\\ 
5 & RU & VS & 13 & VU & RS\\ 
6 & RS & RS & 14 & RU & VU\\ 
7 & VS & VU & 15 & VU & RU\\ 
8 & RS & RU & 16 & VU & VU\\ 
\bottomrule[0.09 em]
\end{tabular} 
\end{table}
For instance, the third column can be interpreted as follows: if the persons in a city all are Rather satisfied about the safety they feel walking alone in the night and Very satisfied about the air quality, the model predicts the satisfaction's composition would be (81.4\% `VS', 4.1\% `RS', 12.7\% `RU' and 1.8\% `VU'). The tenth column means that if the persons in a city all are Very unsatisfied about the safety they feel walking alone in the night and Very satisfied about the air quality, the model predicts the satisfaction's composition would be (21.3\% `VS', 78.7\% `RS', 0\% `RU' and 0\% `VU').
The estimated matrix $\hat{A}$ is most naturally interpreted column-wise, as each column represents a discrete probability distribution over the response categories. From this perspective, the matrix exhibits a coherent and structured pattern, with several relevant features that support both interpretability and model validity.
Early columns are highly concentrated on the first category, with probabilities exceeding $0.70$ and above $0.90$ in the first and second ones. These configurations are therefore strongly associated with low outcome levels (satisfaction), indicating high certainty in the prediction.
Intermediate columns display a shift of mass towards the second category, which often becomes dominant (around $0.60$--$0.80$). These represent a transition region, where the model captures intermediate states of the response variable. Later columns exhibit a substantial increase in higher categories (dissatisfaction). Most columns are characterized by sparse distributions, where one category clearly dominates and the remaining categories have very small or zero probability. This behavior enhances interpretability and reflects the effect of regularization, which prevents overly diffuse or noisy distributions.
Some columns are nearly identical, most notably the numbers eight and nine. These columns share almost the same probability distribution, indicating that different combinations of input components can lead to equivalent predictions. This reflects an inherent symmetry in the tensor product construction of the predictors.
When fixing one dimension and progressively worsening the other, the associated distributions exhibit a gradual and coherent shift of probability mass from lower to higher categories. This behavior confirms that the model preserves the expected ordinal structure, as deteriorations in one component of the predictors translate into systematically worse outcomes in the response distribution.
While the overall structure is coherent, some minor deviations can be observed with small amounts of probability mass occasionally appear in non-dominant categories.
These features reflect the model's ability to capture heterogeneity in the data. Slight departures from strict monotonicity are therefore expected in empirical settings.
Some mild irregularities can be observed in specific configurations, such as column eleven, which remains relatively concentrated in the more positive categories despite representing less favorable conditions. This may appear slightly optimistic at first glance.
However, this pattern can be explained by the nature of the data. The questionnaire from which the data are derived includes multiple items related to overall life satisfaction and the quality of various services and living conditions. As a result, the estimated distributions may reflect not only the specific dimensions considered in the model (e.g., air quality and nighttime safety), but also residual effects of other unobserved factors contributing to overall well-being. In this sense, even when certain dimensions are evaluated negatively, individuals may still report relatively high levels of satisfaction due to compensating aspects of their living environment. The model captures this phenomenon through a partial persistence of probability mass in the higher categories, leading to slightly optimistic distributions in some configurations.

To assess the stability of these estimates, we performed a non-parametric bootstrap using Wasserstein order with $R=2000$ replicates. The bootstrap median and 95\% percentile confidence intervals (inf and sup) for each column are shown in the following matrices $A^{(\textrm{med})}, A^{(\textrm{inf})}, A^{(\textrm{sup})}$ respectively:
\[
\resizebox{\columnwidth}{!}{$
A^{(\textrm{inf})} =
\left[
\begin{array}{cccccccccccccccc}
    0.541 & 0.695 & 0.695 & 0.578 & 
    0.537 & 0.362 & 0.214 & 0.218 &
    0.004 & 0.000 & 0.000 & 0.000 &
    0.000 & 0.000 & 0.000 & 0.000 \\
    0.128 & 0.210 & 0.226 & 0.042 &
    0.359 & 0.609 & 0.057 & 0.782 &
    0.973 & 0.879 & 0.873 & 0.744 &
    0.479 & 0.421 & 0.091 & 0.000 \\
    0.285 & 0.060 & 0.049 & 0.233 &
    0.010 & 0.014 & 0.708 & 0.000 &
    0.023 & 0.121 & 0.123 & 0.220 &
    0.418 & 0.326 & 0.528 & 0.227 \\
    0.046 & 0.035 & 0.030 & 0.047 &
    0.004 & 0.015 & 0.021 & 0.000 &
    0.000 & 0.000 & 0.004 & 0.036 &
    0.103 & 0.253 & 0.381 & 0.773
\end{array}
\right]
$}
\]
\[
\resizebox{\columnwidth}{!}{$
A^{(\textrm{med})} =
\left[
\begin{array}{cccccccccccccccc}
    0.914 & 0.887 & 0.807 & 0.754 & 
    0.683 & 0.591 & 0.460 & 0.387 &
    0.310 & 0.209 & 0.168 & 0.082 &
    0.052 & 0.061 & 0.000 & 0.000 \\
    0.000 & 0.016 & 0.128 & 0.175 &
    0.250 & 0.394 & 0.534 & 0.613 &
    0.659 & 0.786 & 0.778 & 0.833 &
    0.571 & 0.423 & 0.260 & 0.059 \\
    0.061 & 0.070 & 0.047 & 0.059 &
    0.061 & 0.005 & 0.000 & 0.000 &
    0.031 & 0.005 & 0.054 & 0.078 &
    0.324 & 0.332 & 0.389 & 0.358 \\
    0.025 & 0.027 & 0.018 & 0.012 &
    0.006 & 0.010 & 0.006 & 0.000 &
    0.000 & 0.000 & 0.000 & 0.007 &
    0.053 & 0.184 & 0.351 & 0.583
\end{array}
\right]
$}
\]
\[
\resizebox{\columnwidth}{!}{$
A^{(\textrm{sup})} =
\left[
\begin{array}{cccccccccccccccc}
    1.000 & 0.986 & 0.980 & 1.000 & 
    0.866 & 0.847 & 0.742 & 0.536 &
    0.538 & 0.377 & 0.288 & 0.266 &
    0.246 & 0.079 & 0.045 & 0.289 \\
    0.000 & 0.000 & 0.014 & 0.000 &
    0.125 & 0.144 & 0.242 & 0.464 &
    0.432 & 0.623 & 0.694 & 0.670 &
    0.436 & 0.549 & 0.431 & 0.000 \\
    0.000 & 0.008 & 0.000 & 0.000 &
    0.004 & 0.004 & 0.012 & 0.000 &
    0.025 & 0.000 & 0.018 & 0.034 &
    0.273 & 0.365 & 0.387 & 0.313 \\
    0.000 & 0.006 & 0.006 & 0.000 &
    0.005 & 0.005 & 0.004 & 0.000 & 
    0.005 & 0.000 & 0.000 & 0.000 &
    0.045 & 0.007 & 0.137 & 0.398
\end{array}
\right]
$}
\]

Finally, to evaluate the overall performance of the fitted model, we calculated the Wasserstein R-squared  and Order Preservation Index. The resulting values are     $R^2_{W} = 0.4738$ and  $\textrm{OPI}_{\preceq_{W}}=0.8246$.
This result implies that our linear model, based on the satisfaction compositions of night-time safety and air quality, successfully reduces the total ordinal error in predicting the composition of overall satisfaction about living in that city by 47.38\% compared to a baseline model that simply uses the median night-time safety and air quality compositions for all predictions. This indicates a substantial and meaningful goodness-of-fit.

\section{Conclusions and future directions}\label{sec7}
In this paper we introduced a unified, transforma-tion-free regression framework for ordinal compositional data. By shifting the analytical focus from the nominal log-ratio geometry of Aitchison to the metric space of optimal transport, we addressed long-standing methodological challenges in compositional data analysis. Rather than forcing constrained compositional data into an unconstrained Euclidean space---a process that inherently discards ordinal relationships and struggles with boundary zero values---our method adapts  geometry to the native structure of the data. Leveraging the weighted 1-Wasserstein distance and constraining the regression mappings to the space of column-stochastic matrices, the proposed framework seamlessly embeds the sequential progression of the categories directly into the estimation process.
Beyond its theoretical consistency, the practical value of this approach lies in its robust computability and profound interpretability. Formulating the estimation exactly as a Linear Programming (LP) problem guarantees global optimality and computational efficiency, bypassing the convergence issues and local minima often associated with complex maximum likelihood estimations. Furthermore, in an era increasingly dominated by opaque, black-box algorithms, our framework yields an inherently interpretable regression operator: each column of the estimated matrix explicitly represents a conditional probability distribution. As demonstrated in our applications to educational assortative mating and urban satisfaction, this allows domain experts to intuitively read the estimated effects, transitions, and structural zeros.
The introduction of the ordinal tensor product to handle interactions formally generalizes the framework, yet it also highlights stimulating avenues for future research. We conclude by outlining three possible directions:

\begin{enumerate}
    \item \textit{High-dimensional spaces and sparsity.} As the number of covariates increases, the dimension of the ordinal tensor product space grows multiplicatively, potentially leading to a curse of dimensionality. Since our estimation is already cast as an LP, it is naturally suited for the inclusion of sparsity-inducing penalties. Exploring group-LASSO regularization or low-rank tensor approximations within the LP constraints could enable simultaneous parameter estimation and implicit feature selection in high-dimensional settings.
    
    \item \textit{Data-driven Metric Learning.} In the current framework, the weights of the 1-Wasserstein ground metric (representing the cost of transitioning between adjacent categories) are fixed \textit{a priori}. Extending the model to a bi-level optimization problem where these weights are actively learned from the data could reveal latent ordinal structures. In the context of survey data, this would empirically quantify the perceived psychological distances between Likert-scale responses, offering profound sociological insights.
    
    \item \textit{Integration of Continuous Covariates.} The current formulation is intrinsically designed to map compositional (or categorical) inputs to compositional outputs. Accommodating unbounded continuous covariates (e.g., age or income) directly into this LP-based mass-transportation framework requires the development of novel functional mappings that preserve the stochastic constraints on the output space. Solving this would significantly broaden the empirical applicability of the model.
\end{enumerate}
Ultimately, by bridging compositional data analysis with optimal transport geometry, this work not only resolves practical issues related to zero-inflation and ordinality but also establishes a highly flexible mathematical foundation for future developments in constrained regression analysis.

\section*{Appendices}\label{app}

\subsection*{A. \  Metric consistency of the joint optimal vertex}

In this Appendix we prove some properties concerning the consistency of the ordinal tensor product, defined in Section 2 of the paper, with the First-order Stochastic Dominance and the Wasserstein order.   

\begin{lemma}\label{dominance_tensor_product}
Let $P, \tilde{P} \in \Delta_n$ and $Q, \tilde{Q} \in \Delta_h$. If $P \succeq_{FSD} \tilde{P}$ and $Q \succeq_{FSD} \tilde{Q}$, then $P \otimes Q \succeq_{FSD} \tilde{P} \otimes \tilde{Q}$.
\end{lemma}

\begin{proof}
Let $P = (x_1, \dots, x_{n+1})$ and $\tilde{P} = (\tilde{x}_1, \dots, \tilde{x}_{n+1})$ be points in $\Delta_n$. Similarly, let $Q = (y_1, \dots, y_{h+1})$ and $\tilde{Q} = (\tilde{y}_1, \dots, \tilde{y}_{h+1})$ be points in $\Delta_h$. Let $H_k(P \otimes Q) = \sum_{(i,j) \in I_k} x_i y_j$ \ be the joint cumulative distribution function  at index $k$, where $I_k = \{ \sigma(1), \ldots, \sigma(k) \}$ is the set of the $k$ pairs with the lowest joint costs. 
Since the ordinal cost function $\mathcal{C}(i, j) = c_i + c'_j + i\delta + j\eta$ is strictly increasing with respect to both components (as the marginal cumulative costs $c_i$ and $c'_j$ are non-decreasing), if a state $(i, j) \in I_k$, then any state $(r, s)$ with $r \le i$ and $s \le j$ must satisfy $\mathcal{C}(r, s) \le \mathcal{C}(i, j)$. Thus, $(r, s) \in I_k$. This implies that $I_k$ is a \textit{down-set} in the product lattice $\{1, \dots, n+1\} \times \{1, \dots, h+1\}$. We can rewrite the joint sum $H_k(P \otimes Q)$ by grouping terms according to the first index $i$:
\begin{equation*}
H_k(P \otimes Q) = \sum_{i=1}^{n+1} x_i \left( \sum_{j: (i,j) \in I_k} y_j \right) = \sum_{i=1}^{n+1} x_i g(i)
\end{equation*}
where $g(i) = \sum_{j: (i,j) \in I_k} y_j$. Because $I_k$ is a down-set, as $i$ increases, the subset of available indices $j$ for which $(i, j) \in I_k$ shrinks or remains unchanged. Since probabilities are non-negative ($y_j \ge 0$), $g(i)$ is a monotonically non-increasing function of $i$.  By the definition of first-order stochastic dominance, if $P \succeq_{FSD} \tilde{P}$, then for any non-increasing sequence $g(i)$, it holds that $\sum_{i=1}^{n+1} x_i g(i) \ge \sum_{i=1}^{n+1} \tilde{x}_i g(i)$. Applying this property to our expression, we obtain:
\begin{equation*}
H_k(P \otimes Q) \ge \sum_{i=1}^{n+1} \tilde{x}_i g(i) = \sum_{i=1}^{n+1} \tilde{x}_i \left( \sum_{j: (i,j) \in I_k} y_j \right).
\end{equation*}
By swapping the order of summation, we can rewrite the right-hand side grouped by the second index $j$:
\begin{equation*}
\sum_{i=1}^{n+1} \tilde{x}_i \left( \sum_{j: (i,j) \in I_k} y_j \right) = \sum_{j=1}^{h+1} y_j \left( \sum_{i: (i,j) \in I_k} \tilde{x}_i \right) = \sum_{j=1}^{h+1} y_j h(j)
\end{equation*}
where $h(j) = \sum_{i: (i,j) \in I_k} \tilde{x}_i$. Using the exact same down-set property of $I_k$ on the second coordinate, $h(j)$ is a monotonically non-increasing function of $j$. Since $Q \succeq_{FSD} \tilde{Q}$, we apply the first-order stochastic dominance property to the second factor:
\begin{equation*}
\sum_{j=1}^{h+1} y_j h(j) \ge \sum_{j=1}^{h+1} \tilde{y}_j h(j) = \sum_{j=1}^{h+1} \tilde{y}_j \left( \sum_{i: (i,j) \in I_k} \tilde{x}_i \right) = \sum_{(i,j) \in I_k} \tilde{x}_i \tilde{y}_j = H_k(\tilde{P} \otimes \tilde{Q}).
\end{equation*}
Chaining the inequalities together yields $H_k(P \otimes Q) \ge H_k(\tilde{P} \otimes \tilde{Q})$ for all $k$. By definition, this establishes that $P \otimes Q \succeq_{FSD} \tilde{P} \otimes \tilde{Q}$, completing the proof.
\end{proof}

\begin{proposition}\label{isometry-prop}
For any $P \in \Delta_n$ and $Q \in \Delta_h$, as $(\delta, \eta) \to 0$, the weighted 1-Wasserstein distances of $\Delta_{nh+n+h}$, $\Delta_n$, and $\Delta_h$ satisfy:
\begin{equation*}
W_{1, \mathbf{w}}(P \otimes Q, V_1 \otimes U_1) = W_{1, \mathbf{v}}(P, V_1) + W_{1, \mathbf{u}}(Q, U_1),
\end{equation*}
where $V_1=(1,0,\dots,0) \in \Delta_n$ and $U_1 =(1,0,\dots, 0) \in \Delta_h$ are the maximal rank vertices  of the simplices $\Delta_n$ and $\Delta_h$, respectively.
\end{proposition}

\begin{proof}
Let $P = (x_1, \dots, x_{n+1}) \in \Delta_n$ and $Q = (y_1, \dots, y_{h+1}) \in \Delta_h$. By equation 3 in the paper, the marginal distances from the optimal vertices can be expressed as the expected costs $W_{1, \mathbf{v}}(P, V_1) = \sum_{i=1}^{n+1} c_i x_i$ and $W_{1, \mathbf{u}}(Q, U_1) = \sum_{j=1}^{h+1} c'_j y_j$. In the joint simplex $\Delta_{nh+n+h}$, the distance of $P \otimes Q$ from the optimal joint vertex $V_1 \otimes U_1$ is given by:
\begin{equation*}
W_{1, \mathbf{w}}\left(P \otimes Q, V_1 \otimes U_1\right) = \sum_{k=1}^{(n+1)(h+1)} \mathcal{C}(\sigma(k)) z_{\sigma(k)}
\end{equation*}
where the components $z_{\sigma(k)}$ correspond exactly to the tensor product terms $x_i y_j$. As $(\delta, \eta) \to 0$, the refined cost $\mathcal{C}(\sigma(k))$ converges to the purely additive cost $c_i + c'_j$. Thus:
\begin{align*}
W_{1, \mathbf{w}}\left(P \otimes Q, V_1 \otimes U_1\right) &= \sum_{i=1}^{n+1} \sum_{j=1}^{h+1} (c_i + c'_j) x_i y_j = \sum_{i=1}^{n+1} c_i x_i \left( \sum_{j=1}^{h+1} y_j \right) + \sum_{j=1}^{h+1} c'_j y_j \left( \sum_{i=1}^{n+1} x_i \right).
\end{align*}
Since $\sum_{j=1}^{h+1} y_j = 1$ and $\sum_{i=1}^{n+1} x_i = 1$, we obtain:
\begin{equation*}
W_{1, \mathbf{w}}(P \otimes Q, V_1 \otimes U_1) = \sum_{i=1}^{n+1} c_i x_i + \sum_{j=1}^{h+1} c'_j y_j = W_{1, \mathbf{v}}(P, V_1) + W_{1, \mathbf{u}}(Q, U_1).
\end{equation*}
\end{proof}

Proposition \ref{isometry-prop} implies that the ordinal tensor product preserves the Wasserstein order on the primary lexicographical component in the limit. The following Corollary rigorously formalizes this property using the exact distance expansion.

\begin{corollary}
Let $P, \tilde{P} \in \Delta_n$ and $Q, \tilde{Q} \in \Delta_h$. Assuming strict preference on the primary distance to the optimal vertex, if $P \succ_{W} \tilde{P}$ and $Q \succeq_{W} \tilde{Q}$, then as $(\delta, \eta) \to 0$, we have $P \otimes Q \succ_{W} \tilde{P} \otimes \tilde{Q}$.
\end{corollary}

\begin{proof}
According to Definition 1, the primary component of the Wasserstein order $\succeq_W$ is the distance from the optimal vertex. In the joint simplex, the exact distance from the optimal joint vertex $V_1 \otimes U_1$ (which corresponds to the state $(1,1)$ with cost $\mathcal{C}(1,1) = \delta + \eta$) is given by the expected cost shift:
\begin{align*}
\delta_1(P \otimes Q) &= \sum_{i=1}^{n+1} \sum_{j=1}^{h+1} \left( \mathcal{C}(i,j) - \mathcal{C}(1,1) \right) x_i y_j \\
&= \sum_{i=1}^{n+1} \sum_{j=1}^{h+1} \left( c_i + c'_j + (i-1)\delta + (j-1)\eta \right) x_i y_j \\
&= \delta_1(P) + \delta_1(Q) + \delta(\mathbb{E}[I_P] - 1) + \eta(\mathbb{E}[I_Q] - 1)
\end{align*}
where $\mathbb{E}[I_P] = \sum_{i=1}^{n+1} i x_i$ and $\mathbb{E}[I_Q] = \sum_{j=1}^{h+1} j y_j$ denote the expected categorical indices.
The assumptions $P \succ_W \tilde{P}$ and $Q \succeq_W \tilde{Q}$ on the primary distance imply $\delta_1(P) < \delta_1(\tilde{P})$ and $\delta_1(Q) \le \delta_1(\tilde{Q})$. Let $\Delta = (\delta_1(\tilde{P}) + \delta_1(\tilde{Q})) - (\delta_1(P) + \delta_1(Q)) > 0$. 
The exact difference in the joint primary distances expands as:
\begin{equation*}
\delta_1(\tilde{P} \otimes \tilde{Q}) - \delta_1(P \otimes Q) = \Delta + \delta \left( \mathbb{E}[I_{\tilde{P}}] - \mathbb{E}[I_P] \right) + \eta \left( \mathbb{E}[I_{\tilde{Q}}] - \mathbb{E}[I_Q] \right).
\end{equation*}
Since $\Delta > 0$ is independent of $\delta$ and $\eta$, the limit as $(\delta, \eta) \to 0$ of the difference converges exactly to $\Delta > 0$. Thus, for sufficiently small $\delta$ and $\eta$, the perturbation terms are strictly bounded by $\Delta$, and we have $\delta_1(P \otimes Q) < \delta_1(\tilde{P} \otimes \tilde{Q})$. 
Since the primary distance component for $P \otimes Q$ is strictly smaller than that of $\tilde{P} \otimes \tilde{Q}$, it follows immediately from the lexicographical definition that $P \otimes Q \succ_{W} \tilde{P} \otimes \tilde{Q}$.
\end{proof}

The following theorem shows that, under the assumption of First-order Stochastic Dominance, the 1-Wasserstein distance between any two joint compositions decomposes exactly into the sum of the marginal distances.

\begin{theorem}
For any $P, \tilde{P} \in\Delta_{n}$ and $Q, \tilde{Q} \in\Delta_h$ such that $P \succeq_{FSD} \tilde{P}$ and $Q \succeq_{FSD} \tilde{Q}$, as $(\delta, \eta) \to 0$ we have:
\begin{equation*}
W_{1, \mathbf{w}}(P \otimes Q, \tilde{P} \otimes \tilde{Q}) = W_{1, \mathbf{v}}(P, \tilde{P}) + W_{1, \mathbf{u}}(Q, \tilde{Q}).
\end{equation*}
\end{theorem}
\begin{proof}
By Lemma \ref{dominance_tensor_product}, $P \otimes Q \succeq_{FSD} \tilde{P} \otimes \tilde{Q}$. Moreover, the assumption of stochastic dominance implies $F_i(P) \ge F_i(\tilde{P})$, allowing us to compute the 1-Wasserstein distance directly without absolute values:
\begin{equation}\label{formula_dominance1}
 W_{1, \mathbf{v}}(P, \tilde{P}) = \sum_{i=1}^{n} v_i \left( F_i(P) - F_i(\tilde{P}) \right) =W_{1, \mathbf{v}}(\tilde{P}, V_1) - W_{1, \mathbf{v}}(P, V_1)    
\end{equation}
and, analogously,
\begin{gather}
W_{1, \mathbf{u}}(Q, \tilde{Q}) = W_{1, \mathbf{u}}(\tilde{Q}, U_1) - W_{1, \mathbf{u}}(Q, U_1), \label{formula_dominance2} \\
W_{1, \mathbf{w}}(P\otimes Q, \tilde{P} \otimes \tilde{Q}) = W_{1, \mathbf{w}}(\tilde{P} \otimes \tilde{Q}, V_1 \otimes U_1) - W_{1, \mathbf{w}}(P\otimes Q, V_1 \otimes U_1). \label{formula_dominance3}
\end{gather}
Then, using Proposition \ref{isometry-prop} alongside equations \eqref{formula_dominance1}, \eqref{formula_dominance2}, and \eqref{formula_dominance3}, we have:
\begin{align*}
W_{1, \mathbf{w}}(P\otimes Q, \tilde{P}\otimes \tilde{Q}) &= W_{1, \mathbf{w}}(\tilde{P}\otimes \tilde{Q}, V_1\otimes U_1) - W_{1, \mathbf{w}}(P\otimes Q, V_1\otimes U_1) \\
&= W_{1, \mathbf{v}}(\tilde{P}, V_1) + W_{1, \mathbf{u}}(\tilde{Q}, U_1) - \left( W_{1, \mathbf{v}}(P, V_1) + W_{1, \mathbf{u}}(Q, U_1) \right) \\
&= W_{1, \mathbf{v}}(P, \tilde{P}) + W_{1, \mathbf{u}}(Q, \tilde{Q}).
\end{align*}
\end{proof}


\subsection*{B. \ Selection of the regularization parameters $\lambda_1$ and $\lambda_2$}

Two complementary approaches are considered for the selection of the regularization parameters $\lambda_1$ and $\lambda_2$: a Generalized Cross Validation (GCV) criterion based on the Singular Value Decomposition (SVD) of the estimated operator, and a standard Cross-Validation (CV) procedure. The two methods exhibit different behaviors depending on the sample size and the stability of the estimation problem.
The SVD-based approach is computationally efficient and particularly suitable in small-sample settings, where repeated data splitting may introduce high variability and unstable estimates. By exploiting the spectral structure of the estimated matrix, the GCV criterion provides a smooth approximation of model complexity without requiring repeated refitting on different folds. However, since it relies on an approximation of the effective degrees of freedom, it may become less reliable in larger samples where direct predictive validation is feasible.
Conversely, the Cross-Validation approach directly evaluates predictive performance on held-out observations and is therefore generally more robust when a sufficiently large number of observations is available. Nevertheless, in small datasets, repeated splitting may lead to considerable variability in the selected parameters and increased computational cost, especially in the presence of constrained LP optimization problems.
For these reasons, the SVD-based criterion is mainly adopted in low-sample and potentially unstable settings, while Cross-Validation is preferred when the available sample size is sufficiently large to provide stable out-of-sample estimates.
For each pair $(\lambda_1,\lambda_2)$ in a predefined grid, the LP problem is solved to obtain the optimal matrix $A_{\lambda_1,\lambda_2}$ satisfying the simplex and cumulative fitting constraints. Let:
\begin{itemize}
    \item[-] $P_i$ be the observed distribution for observation $i$,
    \item[-] $P'_i$ be the target distribution for observation $i$,
    \item[-] $A \in \mathbb{R}^{(m+1) \times (n+1)}$ be the matrix to optimize, with column-stochasticity constraints  $\sum_{r=1}^{m+1} A_{rj}=1$, $A_{rj}\geq0$.
\end{itemize}
The fitted distribution associated with observation $i$ is given by:
$\widehat{P}^{\prime}_i = A P_i$.
The fitting error is then evaluated through the Residual Wasserstein Dispersion
\begin{equation*}
\mathcal{D}_{\mathrm{res}}
=
\sum_{i=1}^{N}
W_{1, \mathbf{a}}
\left(
A P_i,
P'_i
\right)=  \sum_{i=1}^{N}
\sum_{k=1}^{m}
a_k
\left|
F_k(A P_i)
-
F_k(P'_i)
\right|,  
\end{equation*}
where
$F_k(A P_i)
=
\sum_{j=1}^{n+1}
(P_i)_j
\sum_{r=1}^{k} A_{rj} = \sum_{j=1}^{n+1} (P_i)_j F_{k,j}$.
For the SVD-based criterion, the Singular Value Decomposition of the estimated matrix is computed:
\[
A_{\lambda_1,\lambda_2}
=
U \Sigma V^\top,
\]
where:
\begin{itemize}
    \item[-] $U \in \mathbb{R}^{(m+1)\times (m+1)}$ and $V \in \mathbb{R}^{(n+1)\times (n+1)}$ are orthogonal matrices,
    \item[-] $\Sigma = \mathrm{diag}(\sigma_1,\dots,\sigma_r)$ contains the singular values.
\end{itemize}
The SVD provides a spectral characterization of the effective complexity of the estimated
operator (see \cite{caterina} for related uses of SVD in regularization problems). Large
singular values correspond to dominant structural components, whereas small singular values correspond to directions strongly penalized by the regularization. The effective degrees of freedom are approximated by:
\[
\mathrm{df}_{\lambda_1,\lambda_2}
=
\sum_{i=1}^{r}
\frac{\sigma_i^2}
{\sigma_i^2+\lambda_1+\lambda_2}.
\]
The corresponding GCV criterion is:
\[
\mathrm{GCV}(\lambda_1,\lambda_2)
=
\frac{\mathcal{D}_{\mathrm{res}}}
{\left(N_{\mathrm{cumul}}-\mathrm{df}_{\lambda_1,\lambda_2}\right)^2},
\]
where $N_{\mathrm{cumul}} = N \cdot m$ is the total number of cumulative observations. The selected parameters are therefore:
\[
(\lambda_1^*,\lambda_2^*)
=
\arg\min_{(\lambda_1,\lambda_2)}
\mathrm{GCV}(\lambda_1,\lambda_2).
\]
Alternatively, in the Cross-Validation approach, the observations are partitioned into training and validation folds. For each candidate pair $(\lambda_1,\lambda_2)$, the LP model is estimated on the training folds and the predictive fitting error is evaluated on the validation fold using the same cumulative discrepancy measure. The selected parameters are those minimizing the average out-of-sample validation error across folds. In the numerical experiments, the search is performed over a grid of values in the interval $[0,0.1]$. This relatively small range is sufficient to stabilize the estimated conditional distributions while avoiding excessive oversmoothing, which may otherwise suppress relevant structural variability in the transport map.

\subsection*{C. \ Algorithm Setting I}
Algorithm 1 describes in details the procedure used in simulation Setting 1.
\begin{algorithm}
\caption{Monte Carlo Simulation Comparing WD and COD}\label{alg:simulation}
\begin{algorithmic}[1]
\Require Sample size $N$; number of replications $T$; 
number of categories $n+1, m+1$; 
Dirichlet concentration parameter $\alpha$; 
 $\nu$.

\For{$t = 1$ to $T$}

    \State \textbf{(1) Generate initial sample}
    
    \For{$i = 1,\dots,N$}
        \State Sample $P_i \sim \text{Dirichlet}(1_{n+1})$;
        \State Compute latent score 
        $s_i = \sum_{k=1}^{n+1} k P_{ik}$;
        \State Construct target probabilities $p_i$ according to:
         $p_{ij} \propto \exp(-\nu |j - s_i|)$;
        \State Normalize $p_i$ so that $\sum_j p_{ij}=1$;
        \State Sample $P'_i \sim \text{Dirichlet}(\alpha p_i)$;
    \EndFor
    
    \State \textbf{(2) Estimate structural operators}
    
    \State Estimate $A^{WD}$ via optimal transport;
    \State Estimate $A^{COD}$ via COD algorithm;
    
    \State \textbf{(3) Generate new evaluation sample}
    
    \For{$i = 1,\dots,N$}
        \State Sample $P_i^{new} \sim \text{Dirichlet}(1_{n+1})$;
        \State Compute $P_i^{'WD,*} = A^{WD} P_i^{new}$;
        \State Compute $P_i^{'COD,*} = A^{COD} P_i^{new}$;
        \State Sample
        \[
        P_i^{'WD} \sim \text{Dirichlet}(\alpha P_i^{'WD,*}),
        \quad
        P_i^{'COD} \sim \text{Dirichlet}(\alpha P_i^{'COD,*});
        \]
    \EndFor
    
    \State \textbf{(4) Train/validation split}
    
    \State Randomly assign $70\%$ of observations to the training set;
    
    \State Re-estimate WD and COD on training data;
    
    \State Compute validation errors:
    \begin{itemize}
        \item RPS;
        \item BS.
    \end{itemize}
    
    \State Record mean errors and win indicators.
    
\EndFor

\State Compute average validation errors and win rates across replications.

\end{algorithmic}
\end{algorithm}

For each Monte Carlo replication, we compute binary win indicators
based on pairwise comparisons of validation errors. Let
$\mathcal{E}_{WD}^{(d)}$ and 
$\mathcal{E}_{COD}^{(d)}$
denote the validation errors of the WD and COD estimators,
respectively, under data generating process $d$. In the implementation, three groups of comparisons are recorded.
\[ \textit{(i) Wins under WD-DGP:} \quad
1\!\!1 \big(
\mathcal{E}_{WD}^{(WD)} < \mathcal{E}_{COD}^{(WD)}
\big).
\]
This comparison evaluates both estimators on data generated
from the WD-based DGP. It measures how well each method
recovers the structural operator when the true data generating
mechanism is aligned with the geometry exploited by WD.
If the DGP embeds a metric ordinal structure consistent with
Wasserstein geometry (e.g., ordinality of the support and
distance-based probability decay), WD should win with high probability.
In particular, we expect stronger improvements under the Laplace mechanism and with a larger number of ordinal categories because the exponential decay in absolute distance aligns directly with the linear ground metric of Wasserstein geometry, and a finer ordinal support reduces discretization effects, making the underlying transport structure more identifiable.
\[ \textit{(ii) Wins under COD-DGP:} \quad 
1\!\!1\big(
\mathcal{E}_{WD}^{(COD)} < \mathcal{E}_{COD}^{(COD)}
\big).
\]
Here both estimators are evaluated on data generated
from the COD-based DGP. This comparison measures robustness
when the true structural operator coincides with the COD mechanism.
If the COD DGP does not strongly rely on metric displacement,
COD should dominate. However, if the COD-induced structure
still embeds ordinal regularity, WD may remain competitive.
A win rate for WD substantially above $50\%$ would indicate
robustness of the WD estimator to structural misspecification.
\[\textit{(iii) Wins under “True DGP”:} \quad
1\!\!1\big(
\mathcal{E}_{WD}^{(WD)} < \mathcal{E}_{COD}^{(COD)}
\big).
\]
This is not a symmetric estimator comparison under a fixed DGP.
Instead, it contrasts the performance of each estimator under
its own favorable data generating environment.
In other words, it asks whether: Does WD under its natural DGP outperform COD under its natural DGP?
This comparison measures the relative intrinsic efficiency
of the two estimation paradigms when each is evaluated under
its structurally coherent environment.
However, this comparison is inherently asymmetric and
should not be interpreted as a causal dominance statement,
since the underlying DGP differs across the two sides
of the inequality.

\subsection*{D. \ Simulation results}
Here are additional simulation results for $N=100$ and $N=500$ (see Tables \ref{ta1},\ref{ta2},\ref{ta3},\ref{ta4},\ref{ta5},\ref{ta6}). The parameters used in all cases are $\nu=2$, $\alpha=80$, $w_1=0.6$ and $w_2=0.4$.

\begin{table}[H]
\centering
\caption{Setting 1: win rates of OT method under different DGP ($N=500$).}\label{ta1}
\begin{tabular}{ccccccccc}
\toprule
 &  \multicolumn{2}{c}{WD-DGP} 
 & \multicolumn{2}{c}{COD-DGP}
 & \multicolumn{2}{c}{TRUE-DGP} \\
\cmidrule(lr){2-3}\cmidrule(lr){4-5}\cmidrule(lr){6-7}
$(n+1,m+1)$ & RPS & BS & RPS & BS & RPS & BS \\
\midrule

(3,3) 
       & 0.540 & 0.536 & 0.418 & 0.415 & 0.501 & 0.512 \\

\midrule

(7,7) 
      & 0.774 & 0.835 & 0.655 & 0.718 & 0.759 & 0.746 \\

\midrule

(5,5) 
       & 0.602 & 0.611 & 0.474 & 0.471 & 0.585 & 0.581 \\

\midrule

(3,5) 
        & 0.571 & 0.584 & 0.504 & 0.511 & 0.602 & 0.573 \\

\midrule

(5,3) 
       & 0.586 & 0.599 & 0.369 & 0.400 & 0.598 & 0.586 \\

\midrule

(5,7) 
       & 0.621 & 0.640 & 0.420 & 0.434 & 0.632 & 0.626 \\

\midrule

(7,5) 
      & 0.761 & 0.814 & 0.580 & 0.597 & 0.742 & 0.721 \\

\bottomrule
\end{tabular}
\end{table}

\begin{table}[!ht]
\centering
\caption{Win rates of WD method under different DGP ($N=100$). Zeros are randomly selected in the validation set given a proportion ($p$) of the total number of components.}\label{ta2}
\begin{tabular}{cccccccccc}
\toprule
 & & \multicolumn{2}{c}{WD DGP} 
 & \multicolumn{2}{c}{COD DGP}
 & \multicolumn{2}{c}{TRUE DGP} \\
\cmidrule(lr){3-4}\cmidrule(lr){5-6}\cmidrule(lr){7-8}
$(n+1,m+1)$ & $p$ & RPS & Brier & RPS & Brier & RPS & Brier \\
\midrule

(3,3) & 0 & 0.530 & 0.522 & 0.430 & 0.428 & 0.504 & 0.506 \\  
& 0.05 & 0.526 & 0.480 & 0.450 & 0.428 & 0.494 & 0.478 \\  
& 0.10 & 0.532 & 0.494 & 0.492 & 0.452 & 0.562 & 0.548 \\
& 0.15 &  0.528 & 0.490 & 0.534 & 0.503 & 0.579 & 0.560 \\
& 0.20 &  0.530 & 0.519 & 0.539 & 0.510 & 0.589 & 0.570 \\
\midrule

(7,7) & 0 & 0.742 & 0.772 & 0.654 & 0.714 & 0.744 & 0.726 \\  
& 0.05 & 0.754 & 0.816 & 0.628 & 0.668 & 0.764 & 0.766 \\  
& 0.10 & 0.746 & 0.782 & 0.678 & 0.694 & 0.814 & 0.732 \\
& 0.15 & 0.758 & 0.794 & 0.724 & 0.742 & 0.798 & 0.768 \\
& 0.20 & 0.777 & 0.814 & 0.724 & 0.752 & 0.804 & 0.784 \\
\midrule

(5,5) & 0 & 0.579 & 0.583 & 0.480 & 0.480 & 0.574 & 0.572  \\ 
& 0.05 & 0.584 & 0.622 & 0.506 & 0.538 & 0.602 & 0.592 \\  
& 0.10 & 0.618 & 0.638 & 0.542 & 0.554 & 0.664 & 0.662 \\
& 0.15 & 0.628 & 0.664 & 0.550 & 0.568 & 0.626 & 0.612 \\
& 0.20 & 0.676 & 0.706 & 0.578 & 0.600 & 0.672 & 0.656 \\
\midrule

(3,5) & 0 & 0.538 & 0.546 & 0.526 & 0.524 & 0.598 & 0.568 \\
& 0.05 & 0.530 & 0.541 & 0.478 & 0.514 & 0.588 & 0.570 \\  
& 0.10 & 0.531 & 0.541 & 0.476 & 0.512 & 0.585 & 0.578 \\
& 0.15 & 0.538 & 0.547 & 0.496 & 0.514 & 0.587 & 0.581 \\
& 0.20 & 0.665 & 0.690 & 0.568 & 0.582 & 0.622 & 0.609 \\
\midrule

(5,3) & 0 & 0.558 & 0.576 & 0.382 & 0.412 & 0.592 & 0.580  \\ 
& 0.05 & 0.598 & 0.610 & 0.564 & 0.540 & 0.650 & 0.624 \\  
& 0.10 & 0.614 & 0.640 & 0.622 & 0.600 & 0.676 & 0.638 \\
& 0.15 & 0.672 & 0.678 & 0.688 & 0.652 & 0.702 & 0.688 \\
& 0.20 & 0.684 & 0.690 & 0.696 & 0.676 & 0.704 & 0.690 \\
\midrule

(5,7) & 0 & 0.596 & 0.614 & 0.446 & 0.462 & 0.626 & 0.612\\
& 0.05 & 0.590 & 0.608 & 0.478 & 0.514 & 0.588 & 0.610 \\  
& 0.10 & 0.616 & 0.668 & 0.512 & 0.544 & 0.538 & 0.622 \\
& 0.15 & 0.668 & 0.682 & 0.501 & 0.544 & 0.688 & 0.638 \\
& 0.20 & 0.665 & 0.690 & 0.568 & 0.582 & 0.676 & 0.674 \\
\midrule

(7,5) & 0 & 0.730 & 0.782 & 0.616 & 0.626 & 0.734 & 0.718 \\
& 0.05  & 0.744 & 0.754 & 0.646 & 0.648 & 0.742 & 0.696  \\  
& 0.10  & 0.742 & 0.772 & 0.676 & 0.716 & 0.716 & 0.752 \\  
& 0.15 & 0.786 & 0.808 & 0.708 & 0.742 & 0.774 & 0.744 \\
& 0.20 & 0.792 & 0.828 & 0.746 & 0.772 & 0.802 & 0.780 \\
\bottomrule
\end{tabular}
\end{table}

\begin{table}[!ht]
\centering
\caption{Win rates of WD method under different DGP ($N=100$). Zeros are selected in the largest components in the validation set given a proportion ($p$) of the total number of components.}\label{ta3}
\begin{tabular}{cccccccccc}
\toprule
 & & \multicolumn{2}{c}{WD DGP} 
 & \multicolumn{2}{c}{COD DGP}
 & \multicolumn{2}{c}{TRUE DGP} \\
\cmidrule(lr){3-4}\cmidrule(lr){5-6}\cmidrule(lr){7-8}
$(n+1,m+1)$ & $p$ & RPS & Brier & RPS & Brier & RPS & Brier \\
\midrule

(3,3) & 0 & 0.530 & 0.522 & 0.430 & 0.428 & 0.504 & 0.506 \\  
& 0.05 & 0.528 & 0.515 & 0.482 & 0.448 & 0.532 & 0.518 \\  
& 0.10 & 0.540 & 0.550 & 0.536 & 0.542 & 0.580 & 0.596 \\
& 0.15 &  0.616 & 0.622 & 0.590 & 0.602 & 0.612 & 0.626 \\
& 0.20 &  0.650 & 0.640 & 0.622 & 0.678 & 0.654 & 0.668 \\
\midrule

(7,7) & 0 & 0.742 & 0.772 & 0.654 & 0.714 & 0.744 & 0.726 \\  
& 0.05 & 0.816 & 0.842 & 0.690 & 0.726 & 0.806 & 0.770 \\  
& 0.10 & 0.870 & 0.866 & 0.818 & 0.814 & 0.838 & 0.802 \\
& 0.15 & 0.930 & 0.936 & 0.848 & 0.862 & 0.906 & 0.856 \\
& 0.20 & 0.944 & 0.942 & 0.890 & 0.886 & 0.932 & 0.892 \\
\midrule

(5,5) & 0 & 0.579 & 0.583 & 0.480 & 0.480 & 0.574 & 0.572  \\ 
& 0.05 & 0.672 & 0.704 & 0.548 & 0.550 & 0.652 & 0.632 \\  
& 0.10 & 0.724 & 0.754 & 0.624 & 0.662 & 0.714 & 0.706 \\
& 0.15 & 0.770 & 0.806 & 0.712 & 0.710 & 0.760 & 0.714 \\
& 0.20 & 0.858 & 0.900 & 0.786 & 0.846 & 0.840 & 0.832 \\

\midrule

(3,5) & 0 & 0.538 & 0.546 & 0.526 & 0.524 & 0.598 & 0.568 \\
& 0.05 & 0.548 & 0.510 & 0.506 & 0.518 & 0.610 & 0.588 \\  
& 0.10 & 0.734 & 0.694 & 0.640 & 0.672 & 0.764 & 0.716 \\
& 0.15 & 0.810 & 0.790 & 0.744 & 0.792 & 0.782 & 0.744 \\
& 0.20 & 0.920 & 0.900 & 0.880 & 0.900 & 0.898 & 0.856 \\
\midrule

(5,3) & 0 & 0.558 & 0.576 & 0.382 & 0.412 & 0.592 & 0.580  \\ 
& 0.2 & 0.696 & 0.712 & 0.630 & 0.598 & 0.664 & 0.650 \\  
& 0.4 & 0.792 & 0.812 & 0.748 & 0.730 & 0.750 & 0.738 \\
& 0.6 & 0.908 & 0.894 & 0.870 & 0.818 & 0.868 & 0.842 \\
& 0.6 & 0.950 & 0.960 & 0.928 & 0.900 & 0.926 & 0.930 \\
\midrule

(5,7) & 0 & 0.596 & 0.614 & 0.446 & 0.462 & 0.626 & 0.612\\
& 0.05 & 0.738 & 0.732 & 0.554 & 0.592 & 0.648 & 0.640 \\  
& 0.10 & 0.818 & 0.846 & 0.728 & 0.746 & 0.794 & 0.754 \\
& 0.15 & 0.890 & 0.900 & 0.814 & 0.836 & 0.870 & 0.838 \\
& 0.20 & 0.922 & 0.946 & 0.862 & 0.916 & 0.916 & 0.906 \\
\midrule

(7,5) & 0 & 0.730 & 0.782 & 0.616 & 0.626 & 0.734 & 0.718 \\
& 0.05  & 0.814 & 0.826 & 0.736 & 0.760 & 0.810 & 0.796\\  
& 0.10 & 0.856 & 0.874 & 0.808 & 0.832 & 0.840 & 0.826 \\
& 0.15 & 0.924 & 0.944 & 0.838 & 0.896 & 0.874 & 0.872 \\
& 0.20 & 0.950 & 0.952 & 0.928 & 0.932 & 0.942 & 0.924 \\
\bottomrule
\end{tabular}
\end{table}

\begin{table}[!ht]
\centering
\caption{Win rates of WD method under different DGP ($N=100$) with inversions.}\label{ta4}
\begin{tabular}{cccccccccc}
\toprule
 & & \multicolumn{2}{c}{WD DGP} 
 & \multicolumn{2}{c}{COD DGP}
 & \multicolumn{2}{c}{TRUE DGP} \\
\cmidrule(lr){3-4}\cmidrule(lr){5-6}\cmidrule(lr){7-8}
$(n+1,m+1)$ & Prop inversions & RPS & Brier & RPS & Brier & RPS & Brier \\
\midrule

(3,3) & 0 & 0.530 & 0.522 & 0.430 & 0.428 & 0.504 & 0.506 \\ 
& 0.1 & 0.520 & 0.520 & 0.450 & 0.443 & 0.531 & 0.512 \\ 
& 0.2 & 0.524 & 0.519 & 0.518 & 0.487 & 0.551 & 0.522 \\  
& 0.3 & 0.557 & 0.541 & 0.544 & 0.535 & 0.546 & 0.540 \\
& 0.4 & 0.589 & 0.577 & 0.558 & 0.562 & 0.634 & 0.590\\

\midrule

(7,7) & 0 & 0.742 & 0.772 & 0.654 & 0.714 & 0.744 & 0.726 \\
& 0.1 & 0.752 & 0.780 & 0.692 & 0.724 & 0.762 & 0.740 \\  
& 0.2 & 0.782 & 0.818 & 0.770 & 0.764 & 0.794 & 0.776 \\  
& 0.3 &  0.840 & 0.870 & 0.804 & 0.812 & 0.838 & 0.816 \\
& 0.4 & 0.904 & 0.914 & 0.802 & 0.826 & 0.894 & 0.870 \\

\midrule

(5,5) & 0 & 0.579 & 0.583 & 0.480 & 0.480 & 0.574 & 0.572  \\ 
& 0.1 & 0.600 & 0.596 & 0.488 & 0.532 & 0.602 & 0.610 \\ 
& 0.2 & 0.616 & 0.638 & 0.550 & 0.538 & 0.644 & 0.642 \\  
& 0.3 & 0.664 & 0.694 & 0.588 & 0.612 & 0.730 & 0.700 \\
& 0.4 & 0.744 & 0.766 & 0.652 & 0.686 & 0.720 & 0.726 \\

\midrule

(3,5) & 0 & 0.538 & 0.546 & 0.526 & 0.524 & 0.598 & 0.568 \\
& 0.1 & 0.634 & 0.550 & 0.542 & 0.498 & 0.690 & 0.576 \\ 
& 0.2 & 0.830 & 0.636 & 0.742 & 0.600 & 0.818 & 0.674 \\  
& 0.3 & 0.946 & 0.770 & 0.914 & 0.784 & 0.924 & 0.762 \\
& 0.4 & 0.988 & 0.872 & 0.966 & 0.892 & 0.970 & 0.850 \\

\midrule

(5,3) & 0 & 0.558 & 0.576 & 0.382 & 0.412 & 0.592 & 0.580  \\ 
& 0.1 & 0.704 & 0.666 & 0.688 & 0.600 & 0.696 & 0.648 \\
& 0.2 & 0.864 & 0.806 & 0.832 & 0.726 & 0.828 & 0.752 \\ 
& 0.3 & 0.902 & 0.882 & 0.910 & 0.824 & 0.914 & 0.846 \\
& 0.4 & 0.982 & 0.962 & 0.974 & 0.916 & 0.974 & 0.912 \\

\midrule

(5,7) & 0 & 0.596 & 0.614 & 0.446 & 0.462 & 0.626 & 0.612\\
& 0.1 & 0.686 & 0.664 & 0.574 & 0.588 & 0.710 & 0.638 \\ 
& 0.2 & 0.816 & 0.762 & 0.726 & 0.672 & 0.848 & 0.736 \\  
& 0.3 & 0.934 & 0.868 & 0.878 & 0.794 & 0.924 & 0.806 \\
& 0.4 & 0.978 & 0.924 & 0.946 & 0.888 & 0.950 & 0.884 \\

\midrule

(7,5) & 0 & 0.730 & 0.782 & 0.616 & 0.626 & 0.734 & 0.718 \\
& 0.1  & 0.842 & 0.814 & 0.766 & 0.744 & 0.820 & 0.770\\  
& 0.2  & 0.896 & 0.852 & 0.844 & 0.788 & 0.882 & 0.828\\  
& 0.3 & 0.948 & 0.908 & 0.922 & 0.890 & 0.938 & 0.872 \\
& 0.4 & 0.986 & 0.956 & 0.958 & 0.942 & 0.982 & 0.932 \\

\bottomrule
\end{tabular}
\end{table}

\begin{table}[!ht]
\centering
\caption{Predictive performance of $n\mathcal{D}_{\textrm{res}}$, $\textrm{OPI}_{\preceq_{W}}$ and $R^2_W$, $N=500$, Setting 2, case I (standard dev. in parentheses) }\label{ta5}
\begin{tabular}{ccccc}
\toprule
$(n+1,m+1)$ & $R^2_W$ & $n\mathcal{D}_{\textrm{res}}$  & $\textrm{OPI}_{\preceq_{W}}$\\
\midrule

(3,3) 
 & 0.443 (0.019) & 0.075 (0.002) & 0.938 (0.004)  \\

\midrule

(7,7)
 & 0.281 (0.006) & 0.068 (0.002) & 0.925 (0.008)  \\

\midrule

(5,5)
 & 0.357 (0.012) & 0.071 (0.001) & 0.936 (0.006)  \\

\midrule

(3,5)
 & 0.441 (0.016) & 0.043 (0.001) &  0.938 (0.004)  \\

\midrule

(5,3)
 & 0.340 (0.027) & 0.069 (0.004) & 0.780 (0.012) \\

\midrule

(5,7)
 & 0.351 (0.010) & 0.049 (0.001) & 0.938 (0.007)  \\

\midrule

(7,5)
 & 0.286 (0.007) & 0.088 (0.002) & 0.915 (0.008) \\

\bottomrule
\end{tabular}
\end{table}

\begin{table}[!ht]
\centering
\caption{Predictive performance of $n\mathcal{D}_{\textrm{res}}$, $\textrm{OPI}_{\preceq_{W}}$ and $R^2_W$, $N=500$, Setting 2, case III (standard dev. in parentheses)}\label{ta6}
\begin{tabular}{ccccc}
\toprule
$(n+1,h+1,m+1)$ & $R^2_W$ & $n\mathcal{D}_{\textrm{res}}$ & $\textrm{OPI}_{\preceq_{W}}$  \\
\midrule

(5,5,5) 
 & 0.503 (0.017) & 0.041 (0.001) &  0.919 (0.005)  \\

\midrule

(4,4,4)
& 0.595 (0.010) & 0.043 (0.001) & 0.919 (0.004) \\

\midrule

(7,7,7)
& 0.440 (0.006) & 0.038 (0.001) & 0.931 (0.005) \\

\midrule

(3,5,7)
& 0.595 (0.020) & 0.021 (0.001) & 0.926 (0.019) \\

\midrule

(7,3,5)
& 0.440 (0.011) & 0.052 (0.001) & 0.903 (0.006)\\

\bottomrule
\end{tabular}
\end{table}



\subsection*{E. \  Fréchet mean for the 1-Wasserstein distance on the simplex}

As discussed in Section 3.6 of the paper, the baseline model for the $R^2_W$ index relies on the component-wise empirical median of the cumulative distribution functions. While the equivalence between the 1-Wasserstein distance and the $L_1$ norm of CDFs is a known property for one-dimensional measures \citep{Vallander1973, Santambrogio2015}, here we provide a rigorous constructive proof tailored to our discrete, weighted compositional framework.

\begin{proposition}\label{prop:frechet_median}
Let ${P}'_1, \dots, {P}'_N \in \Delta_m$ be  observed ordinal compositions. The empirical Fréchet mean with respect to the weighted 1-Wasserstein distance $W_{1,\mathbf{a}}$, defined as
$\bar{{P}}' = \underset{Q \in \Delta_m}{\arg\min} \sum_{i=1}^N W_{1,\mathbf{a}}(Q, {P}'_i)$,
is uniquely characterized (up to the standard non-uniqueness of the sample median for even $N$) by the composition whose cumulative distribution function is the pointwise median of the observed CDFs:
\begin{equation*}
\bar{F}_k = \mathrm{median}\left\{ F_k({P}'_1), \ldots, F_k({P}'_N) \right\}, \quad k\in\left\{1,\ldots,m\right\}.
\end{equation*}
\end{proposition}
\begin{proof}
By the definition of the weighted 1-Wasserstein distance on the ordered simplex, the objective function to minimize over $Q \in \Delta_m$ is the sum of absolute deviations between CDFs:
{\small
\begin{equation*}
\mathcal{L}(Q) = \sum_{i=1}^N W_{1,\mathbf{a}}(Q, {P}'_i) = \sum_{i=1}^N \sum_{k=1}^m a_k \left| F_k(Q) - F_k({P}'_i) \right|=\sum_{k=1}^m a_k \left( \sum_{i=1}^N \left| F_k(Q) - F_k({P}'_i) \right| \right).
\end{equation*}
}
Because the function is now strictly additive and separable across the $m$ cumulative dimensions, and the weights $a_k$ are strictly positive constants, the total sum $\mathcal{L}(Q)$ achieves its unconstrained global minimum if and only if each independent term $\sum_{i=1}^N \left| F_k(Q) - F_k({P}'_i) \right|$ is minimized marginally for each $k\in\left\{1,\ldots,m\right\}$. A well-known property of the $L_1$ norm is that a number  $x\in \mathbb{R}$ minimizing $\sum_{i=1}^N |x - c_i|$ from a set of real numbers $\{c_i\}$ is exactly their sample median. Therefore, the unconstrained independent minimizers are exactly $\bar{F}_k = \mathrm{median}( \{ F_k({P}'_i) \}_{i=1}^N )$. To  prove that this sequence defines a valid composition $\bar{{P}}' \in \Delta_m$, it must satisfy the monotonicity property of a CDF: $0 \le \bar{F}_1 \le \dots \le \bar{F}_m \le 1$. By the definition of the input compositions, for every observation $i\in\left\{1,\ldots,N\right\}$ we have $0 \le F_k({P}'_i) \le F_{k+1}({P}'_i) \le 1$. Because the sample median is a strictly order-preserving operator, applying it to both sides of the pointwise inequalities we get, for each $k\in\left\{1,\ldots,m\right\}$,
\begin{equation*}
\bar{F}_k=\mathrm{median}\left( \{ F_k({P}'_i) \}_{i=1}^N \right) \le \mathrm{median}\left( \{ F_{k+1}({P}'_i) \}_{i=1}^N \right) =\bar{F}_{k+1}.
\end{equation*}
Thus, the unconstrained pointwise optimum is exactly the constrained global minimum. The  Fréchet mean composition $\bar{{P}}'=(\bar{x}_{1},\ldots,\bar{x}_{m+1})$ is then trivially recovered by differencing, $   \bar{x}'_k = \bar{F}_k - \bar{F}_{k-1}$ 
(with the convention $\bar{F}_0=0$ and $\bar{F}_{m+1}=1$), completing the proof.
\end{proof}

\clearpage

\bibliographystyle{plainnat}
\bibliography{ref}

\end{document}